\documentclass[a4paper,onecolumn,11pt,accepted=2020-05-09]{quantumarticle}

\pdfoutput=1

\usepackage[T1]{fontenc}
\usepackage{palatino} 
\usepackage{braket}
\usepackage{amsfonts}
\usepackage{amssymb}
\usepackage{amsmath}
\usepackage{bbm}
\usepackage{latexsym}
\usepackage{amsthm}
\usepackage[usenames]{color}
\usepackage{hyperref}
\usepackage{diagbox}
\usepackage{verbatim}
\usepackage{indentfirst}
\usepackage{tikz-cd}
\usetikzlibrary{decorations.pathmorphing}
\usetikzlibrary{arrows.meta}
\usepackage{tabularx}

\hypersetup{
colorlinks = true
}

\hypersetup{pdfpagemode=UseNone}

\newtheorem{theorem}{Theorem}
\newtheorem{lemma}[theorem]{Lemma}

\newtheorem{protocol}{Protocol}
\newtheorem{remark}{Remark}

\newcommand{\norm}[1]{\ensuremath{\left\lVert #1 \right\rVert}}
\newcommand{\tr}{\operatorname{Tr}}

\newcommand{\eps}{\varepsilon}

\newcommand{\calK}{\mathcal{K}}
\newcommand{\calA}{\mathcal{A}}

\newcommand{\calG}{\mathcal{G}}
\newcommand{\calJ}{\mathcal{J}}
\newcommand{\calL}{\mathcal{L}}

\newcommand{\calS}{\mathcal{S}}

\newcommand{\calX}{\mathcal{X}}
\newcommand{\calY}{\mathcal{Y}}
\newcommand{\calZ}{\mathcal{Z}}

\newcommand{\calH}{\mathcal{H}}
\newcommand{\calR}{\mathcal{R}}

\newcommand{\R}{\mathbb{R}}

\newcommand{\Id}{\mathbbm{1}}

\newcommand{\Herm}{\mathrm{Herm}}
\newcommand{\Pos}{\mathrm{Pos}}

\newcommand{\xot}{\mathrm{XOT}}
\newcommand{\ms}{\mathrm{MS}}

\newcommand{\acc}{\mathrm{acc}}
\newcommand{\rej}{\mathrm{rej}}
\newcommand{\test}{\mathrm{test}}
\newcommand{\prot}{\mathrm{prot}}
\newcommand{\jk}{\mathrm{junk}}
\newcommand{\rL}{\mathrm{L}}
\newcommand{\proj}{\Pi}
\newcommand{\povm}{E}
\newcommand{\meas}{M}

\newcommand{\ketbra}[2]{\ket{#1}\!\!\bra{#2}}
\newcommand{\kb}[1]{\ketbra{#1}{#1}}

\newcommand{\PB}{\mathrm{P}^{\xot}_{B}} 
\newcommand{\PA}{\mathrm{P}^{\xot}_{A}} 

\newcommand{\X}{X}
\newcommand{\Y}{Y}
\newcommand{\A}{A}
\newcommand{\B}{B}
\newcommand{\ops}{O}
  
\begin{document}

\title{A device-independent protocol for XOR oblivious transfer}

\author{Srijita Kundu}
\affiliation{Institute for Quantum Computing, University of Waterloo, Waterloo, Ontario, Canada}
\email{srijita.kundu@uwaterloo.ca}
\orcid{0000-0002-8630-0113}

\author{Jamie Sikora}
\affiliation{Perimeter Institute for Theoretical Physics, Waterloo, Ontario, Canada and Virginia Polytechnic Institute and State University, Blacksburg, Virginia, USA}
\email{sikora@vt.edu}

\author{Ernest Y.-Z. Tan}
\affiliation{Institute for Quantum Computing, University of Waterloo, Waterloo, Ontario, Canada}
\email{yzetan@uwaterloo.ca}
\orcid{0000-0003-4872-158X}

\maketitle 

\begin{abstract} 
Oblivious transfer is a cryptographic primitive where Alice has two bits and Bob wishes to learn some function of them. 
Ideally, Alice should not learn Bob's desired function choice and Bob should not learn any more than what is logically implied by the function value. 
While decent quantum protocols for this task are known, many become completely insecure if an adversary were to control the quantum devices used in the implementation of the protocol. 
In this work we give a fully device-independent quantum protocol for XOR oblivious transfer.
\end{abstract} 

\section{Introduction}  
Oblivious transfer is an important cryptographic primitive in two-party computation as it can be used as a universal building block for constructing more elaborate protocols~\cite{Kil88}. 
Indeed, some quantum protocols for this task are known~\cite{WST08, Schaff07, CKS13, CGS16}. 
However, it can be shown that there do not exist classical protocols with any level of information-theoretic security, and there do not exist quantum protocols with perfect security~\cite{CKS13, Low97}. 

In this paper, we consider a variant of oblivious transfer called \emph{XOR oblivious transfer (XOT)}. 
This is the two-party cryptographic primitive in which two spatially separated parties, Alice and Bob, wish to do the following task: Alice outputs two bits $(x_0, x_1)$, which are uniformly random, and Bob outputs a trit $b$ which is uniformly random, as well as 
$x_b$ where we define $x_2 = x_0 \oplus x_1$. 
In other words, Alice and Bob communicate and Bob learns one bit of information from Alice's two bits (either the first bit, second bit, or their XOR).  
When designing XOT protocols, the security goals are: 
\begin{enumerate}
  \item \emph{Completeness:} If both parties are honest, then their outcomes are consistent (i.e., $x_b$ is the correct value), uniformly random, and neither party aborts. 
  \item \emph{Soundness against cheating Bob:} If Alice is honest, then a dishonest (i.e., cheating) Bob cannot learn both $x_0$ and $x_1$ by digressing from protocol. 
  \item \emph{Soundness against cheating Alice:} If Bob is honest, then a dishonest (i.e., cheating) Alice cannot learn $b$ by digressing from protocol.   
\end{enumerate}

\begin{remark} 
One could imagine a situation where Alice already has a fixed choice of $(x_0, x_1)$ that she wishes to input into a XOT protocol (perhaps from the result of an earlier computation). 
However, we can use the outcomes of an XOT protocol as described above as a one-time pad to convey the information to Bob. 
For more details, see~\cite{CKS13}. 
\end{remark} 

\begin{figure}[!h]
\centering
\begin{tikzpicture}
\node at (0.4,-0.4) {\includegraphics[scale=0.45]{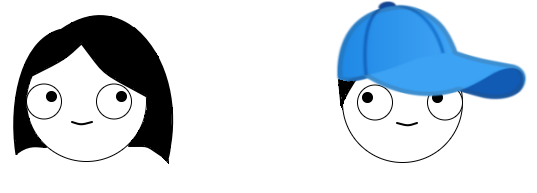}};
\draw[rounded corners] (-2.5, -1.4) rectangle (2.5, -2.8) node[pos=0.5] {XOR Oblivious Transfer};
\draw[-Stealth] (-2, -2.8) -- (-2, -3.6);
\node at (-1.5,-3.8) {$(x_0,x_1, {\color{blue} x_2 = x_0\oplus x_1})$};
\draw[-Stealth] (2,-2.8) -- (2,-3.6);
\node at (1.8, -3.8) {$(b, x_b)$};
\end{tikzpicture}
\caption{Desired outputs for XOR oblivious transfer (XOT).}
\end{figure}
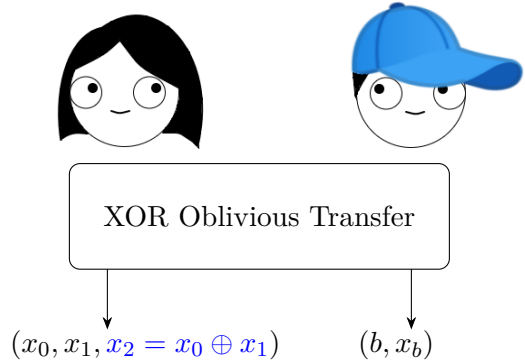

In this paper we are concerned with \emph{information-theoretic security}, meaning that Alice and Bob are only bounded by the laws of quantum mechanics. 
In other words, Alice and Bob can perform arbitrarily complicated computations, have arbitrarily large quantum memories, and so on. 
We shall have occasion to change how much control Alice and Bob have over the protocol, but precisely what actions are allowed to be performed by dishonest parties should be clear from context, and will be described shortly. 

We focus on studying XOT protocols from the perspective of assuming perfect completeness and trying to make them as sound as possible.\footnote{To contrast, the task of finding protocols with perfect soundness and the best possible completeness was considered in \cite{SCK14}.}
To this end, we choose to quantify the soundness via \emph{cheating probabilities}, which we define as follows:
\begin{center} 
\begin{tabularx}{\textwidth}{rX}
$\PB$: & The maximum probability with which a dishonest Bob can learn both of honest Alice's outcome bits $(x_0, x_1)$ and the protocol does not abort. \\ 
$\PA$: & The maximum probability with which a dishonest Alice can learn honest Bob's choice outcome $b$ and the protocol does not abort. \\ 
\end{tabularx} 
\end{center} 
Any XOT protocol with perfect completeness necessarily has $\PA \geq \frac{1}{3}$ and $\PB \geq \frac{1}{2}$,
since a dishonest Alice could always guess one of three choices for Bob's outcome $b$ uniformly at random and, similarly, dishonest Bob can follow the honest protocol to gain perfect knowledge of $x_0$, $x_1$, or $x_0 \oplus x_1$, and then randomly guess the unknown bit in Alice's outcome $(x_0, x_1)$.

\begin{remark}
In this work, we chose to quantify soundness via cheating probabilities, but we note that such a measure of security is not necessarily \emph{composable}~\cite{WW08,VPdR19}.  
Unfortunately, it can be very challenging to prove that a protocol is composably secure, and in some settings such protocols are in fact impossible~\cite{VPdR19}. 
As a first analysis of the protocol proposed in this work, we will restrict ourselves to studying the cheating probabilities only.
\end{remark}
   
Since the lowest possible bounds on $\PA$ and $\PB$ for perfectly complete XOT protocols are asymmetric, we shall consider them in pairs and will not concern ourselves with finding an ``optimal protocol''. 
Instead, we motivate our work by asking the following question: 
\begin{flushleft}
\emph{``Is is possible to create quantum protocols where both $\PA, \PB < 1$ when Alice and Bob do not even trust their own quantum devices?''} 
\end{flushleft}  

When taken literally, this question cannot be answered in the affirmative, since arbitrarily malicious devices could simply broadcast all desired information to a dishonest party. 
However, it turns out that there exist quantum protocols that can be proven secure using almost no assumptions other than ruling out this extreme scenario (which seems a rather necessary assumption in any case). 
This is the notion of \emph{device-independent security}, which typically exploits \emph{nonlocal games} played using entangled states. In a fully device-independent model, one only assumes that the parties' devices do not directly broadcast certain information to the dishonest party and/or each other (we shall explain this in more detail in Section~\ref{sect:DI}). 
In particular, one does not assume that the states and/or measurements implemented by the devices are known, and even the dimensions of the quantum systems are not specified.
Device-independent security analyses exist for other cryptographic tasks such as quantum key distribution~\cite{PAB+09,ADF+18}, bit commitment~\cite{SCA+11, AMPS16}, coin-flipping~\cite{ACK+14}, and weak string erasure (in the bounded or noisy storage models)~\cite{KW16}. 
The approach we use in this work is somewhat similar to~\cite{ACK+14}.

Coin flipping, bit commitment and weak string erasure are all two-party cryptographic tasks like oblivious transfer. There are classical constructions of bit commitment and coin-flipping from oblivious transfer. Since these constructions do not use trusted quantum states or measurements to perform the conversion, we think it would be possible to obtain device-independent bit commitment and coin-flipping from device-independent oblivious transfer, although the resulting cheating probabilities may not be optimal. In the reverse direction, there is a quantum construction of oblivious transfer from bit commitment as well \cite{Unr10}. However, this construction was not designed to be device-independent (it relies on sending trusted states), hence a device-independent protocol for oblivious transfer is not implied by \cite{SCA+11, AMPS16}. Weak string erasure is a task very similar to oblivious transfer, but the protocol for weak string erasure given in \cite{KW16} assumes noisy or bounded quantum storage in addition to the standard minimal assumptions in the device-independent framework, and hence does not directly yield device-independent oblivious transfer without those assumptions.

In this work, we examine the security of quantum XOT protocols in semi-device-independent and device-independent scenarios. A semi-device-independent scenario is weaker than the fully device-independent scenario described above; specifically, the honest parties may trust the state preparation
devices, or the measurement devices, or quantum operations, or some combination of these. By a slight abuse of notation, we use the same notation $\PA$ and $\PB$ to denote the cheating probabilities of Alice and Bob in all the different scenarios, corresponding to differently defined cheating capabilities of the dishonest party. 
For example, if we were to allow a dishonest Alice to control Bob's measurements, it may lead to a different value of $\PA$. 
The cheating capabilities of cheating parties should be clear from context when we discuss $\PA$ and $\PB$.   

\subsection{Trivial protocols}  
      
For readers new to oblivious transfer, we present two bad classical protocols and a simple quantum protocol. 

\begin{protocol}[Bad XOT Protocol 1] 
\label{prot:Bad1} 
\hfill
\begin{enumerate}
\item Alice chooses $(x_0, x_1)$ uniformly at random and sends $(x_0, x_1)$ to Bob. 
\item Bob chooses $b$ uniformly at random. 
\item Alice outputs $(x_0, x_1)$ and Bob outputs $(b, x_b)$.
\end{enumerate} 
\end{protocol}     

A moment's thought shows that Bob has full information (he clearly learns $(x_0, x_1)$) while Alice has no information.  
Therefore, we have 
\begin{equation} 
\PA = 1/3 
\quad \text{ and } \quad 
\PB = 1
\end{equation} 
which is as insecure concerning cheating Bob as possible. 

\begin{protocol}[Bad XOT Protocol 2] 
\label{prot:Bad2} 
\hfill
\begin{enumerate}
\item Bob chooses $b$ uniformly at random and sends $b$ to Alice.   
\item Alice chooses $(x_0, x_1)$ uniformly at random and sends $x_b$ to Bob. 
\item Alice outputs $(x_0, x_1)$ and Bob outputs $(b, x_b)$. 
\end{enumerate} 
\end{protocol}   
Here Alice has full information while Bob has none. 
Therefore, we have 
\begin{equation} 
\PA = 1 
\quad \text{ and } \quad 
\PB = 1/2. 
\end{equation} 

\begin{remark} 
These protocols are obviously device-independent since they are classical. But surprisingly, these protocols can be useful --- we shall use them to construct our fully device-independent protocol. 
The idea is this: suppose Alice wishes to test Bob to see if he has been cheating, and aborts if and only if the test fails. 
Then if the test passes, the parties need a way to finish executing the protocol, which they could do using Protocol~\ref{prot:Bad1} (which is independent of previous steps in the tested protocol). 
\end{remark} 
  
\subsection{A quantum protocol for XOT with no device-independent security}
  
Here we present the oblivious transfer protocol from~\cite{CKS13} adapted to the XOT setting. 
For $b \in \{ 0, 1, 2 \}$, let $\ket{\psi^{\pm}_b} \in \calX \calY$ denote the following two-qutrit state: 
\begin{equation} \label{qutrits}
\ket{\psi^{\pm}_b} 
= 
\begin{cases} 
\quad \frac{1}{\sqrt{2}}(\ket{00}_{\calX\calY} \pm \ket{22}_{\calX\calY}) & \text{ if } b = 0, \\ 
\quad \frac{1}{\sqrt{2}}(\ket{11}_{\calX\calY} \pm \ket{22}_{\calX\calY}) & \text{ if } b = 1, \\ 
\quad \frac{1}{\sqrt{2}}(\ket{00}_{\calX\calY} \pm \ket{11}_{\calX\calY}) & \text{ if } b = 2. 
\end{cases} 
\end{equation}   
Note that for every $b \in \{ 0, 1, 2 \}$, we have that 
$\ket{\psi^+_b}$ and $\ket{\psi^-_b}$ are orthogonal.
We are now ready to state the protocol.   
   
\begin{protocol}
\label{prot:non-di-xot}
\hfill
\begin{enumerate}
\item Bob chooses $b \in \{0,1,2\}$ uniformly at random, prepares the state $\ket{\psi^+_b}$ in registers $\calX\calY$, and sends the register $\calX$ to Alice.
\item Alice chooses $(x_0,x_1)$ uniformly at random, performs the unitary
\begin{equation} 
U_{(x_0,x_1)} = (-1)^{x_0}\kb{0} + (-1)^{x_1}\kb{1} + \kb{2} 
\end{equation} 
on $\calX$, and then sends it back to Bob.
\item Bob performs the $2$-outcome measurement 
$\{ \kb{\psi_b^+}, \Id - \kb{\psi_b^+}\} $
and records his outcome as $c = 0$ if he gets $\kb{\psi_b^+}$ and $c=1$ otherwise. 
\item Alice outputs $(x_0,x_1)$ and Bob outputs $(b, c)$. 
\end{enumerate}
\end{protocol} 
Protocol~\ref{prot:non-di-xot} can be checked to be complete (i.e., Bob always gets the correct outcome). 
The cheating probabilities in this protocol in the cases of trusted and untrusted devices are given by Theorem~\ref{th:Oldprotocol} below. We give a proof for the trusted case in Section~\ref{sect:proofsDD}, and the relatively simple proof for the untrusted case is given below.

\begin{theorem} 
\label{th:Oldprotocol}
In Protocol~\ref{prot:non-di-xot}, the cheating probabilities are as listed in the following table. 
(In the untrusted setting, Alice controls Bob's state preparation and Bob controls Alice's unitary.) 
\begin{table}[!h]
\centering
\begin{tabular}{|c|c|c|}
\hline
\rule{0pt}{2.5ex} & $\PA$ & $\PB$ \\[0.5ex]
\hline
\rule{0pt}{2.5ex} Trusted devices 
& 1/2
& 3/4  
\\[0.5ex]
\hline
\rule{0pt}{2.5ex} Untrusted devices 
& 1 
& 1 
\\[0.5ex]
\hline
\end{tabular} 
\end{table} 
\end{theorem} 

\paragraph{Untrusted device cheating strategies.}    
As indicated in the second row of the above table, the protocol is not secure when the devices are untrusted. To see this,
assume Bob has full control over Alice's unitary. 
It could be a unitary which implements a superdense coding protocol: 
\begin{equation} 
U^{cheat}_{(x_0,x_1)} = 
\begin{cases} 
\quad \Id_{\calX} & \text{ if } (x_0, x_1) = (0, 0), \\ 
\quad \sigma_X & \text{ if } (x_0, x_1) = (0, 1), \\ 
\quad \sigma_Y & \text{ if } (x_0, x_1) = (1, 0), \\ 
\quad \sigma_Z & \text{ if } (x_0, x_1) = (1, 1),  \\ 
\end{cases}  
\end{equation} 
where $\sigma_X$, $\sigma_Y$, and $\sigma_Z$ are the Pauli operators.
(Note that this unitary acts on a qubit, which we can assume Bob sends if he wishes, or just define the unitary to act trivially on the $\ket{2}$ subspace.) 
Now, if Bob creates $\ket{\psi^+_2}$ at the beginning of the protocol, he is left with a state from the Bell basis at the end of the protocol. 
In other words, he can perfectly learn $(x_0, x_1)$. 

In the case of cheating Alice, she can simply control Bob's state preparation device to prepare the state $\ket{b}\ket{0} \in \calX \calY$ on input $b$, and have Bob send the register $\calX$. 
From this state, Alice simply measures it to learn $b$. 
Thus, we have 
\begin{equation} 
\PA = 1 
\quad \text{ and } \quad 
\PB = 1 
\quad 
\text{[Devices are NOT trusted].}  
\end{equation}    
 
Thus, we need a clever way to design protocols where Alice and Bob cannot cheat 
in the above fashion.  
This motivates the need for device-independent XOT protocols and some protocol design ideas that should be avoided. 

\subsection{A semi-device-independent XOT protocol from magic square}
\label{intro:ms}
  
Similar to device-independent protocols which exist for other cryptographic tasks, we design our protocols using nonlocal games. 
In this work, we shall make use of the nonlocal game known as the Mermin-Peres magic square game. 
In the magic square game, 
\begin{itemize}
\item Alice and Bob receive respective inputs $a \in \{0,1,2\}$ and $b \in \{0,1,2\}$ independently and uniformly at random.
\item Alice outputs three bits $(x_0,x_1,x_2) \in \{0,1\}^3$ such that $x_0\oplus x_1\oplus x_2 = 0$ and Bob outputs three bits $(y_0,y_1,y_2) \in \{0,1\}^3$ such that $y_0\oplus y_1\oplus y_2 = 1$.
\item Alice and Bob win the game if $x_b = y_a$.
\end{itemize}

\begin{remark}
We have defined the magic square game such that Alice and Bob are \emph{required} to produce bits satisfying their respective parity conditions, rather than defining these parity conditions as part of the win condition. This means that we can equivalently view this as Alice and Bob each producing only the bits $(x_0,x_1)$ and $(y_0,y_1)$ respectively, with the third bit being completely specified by the parity conditions.
\end{remark}

If Alice and Bob are allowed only classical strategies (using e.g.~shared randomness), the magic square game cannot be won with probability greater than ${8}/{9}$. 
However, there exists a quantum strategy, where Alice and Bob share prior entanglement, which wins the magic square game with probability $1$.  
We shall refer to this strategy as the \emph{magic square strategy} which is detailed in Section~\ref{sect:ms-intro}.  

Now, suppose Alice and Bob play the magic square game according to the above description. 
Notice that $x_2$ will always be equal to $x_0 \oplus x_1$, similar to the definition of XOT, and that $(x_0, x_1)$ is uniformly distributed (see Section~\ref{sect:ms-intro}).  
Also, for each of Bob's input choices, he learns either $x_0$, $x_1$, or $x_2 = x_0 \oplus x_1$ depending on the choice of input $a$ for Alice. 
Since $a$ is chosen uniformly at random, this is almost a proper XOT protocol (putting aside soundness for now). 
The only missing ingredient is that Bob knows he has $x_b$, but does not know which of the bits of $(y_0, y_1, y_2)$ it is. 
To fix this small issue, we simply have Alice tell Bob which bit it is. 

We formalize this protocol below and add in a \emph{test} step that helps to prevent cheating. Strictly speaking, Protocol~\ref{prot:XOTMS} should be thought of as a protocol framework, as we have not specified who creates the entangled state that Alice and Bob share --- either party can. 
We consider different security analyses of Protocol~\ref{prot:XOTMS}, corresponding to each of these different cases. 
In the trusted state analysis, the honest party (whomever that may be) creates the state, and in the untrusted state analysis it is the cheating party who does so. 

\begin{protocol}[Quantum XOT protocol from the magic square game] 
\label{prot:XOTMS} 
\hfill
\begin{enumerate}
\item Alice and Bob share the bipartite state used in the magic square strategy. 
\item \label{step:bob-y} Bob chooses $b \in \{ 0, 1, 2 \}$ uniformly at random, performs the measurements corresponding to $b$ in the magic square strategy on his state to get the outcome $(y_0,  y_1 ,
{y_2=y_0 \oplus y_1 \oplus 1}
)$, and sends $(y_0, y_1, y_2)$ to Alice. 
\item \label{step:alice-input} Alice chooses $a \in \{ 0, 1, 2 \}$ uniformly at random and sends $a$ to Bob. 
\item Alice performs the measurement corresponding to $a$ in her magic square strategy on her state to get the outcome $(x_0, x_1, 
{x_2=x_0 \oplus x_1}
)$. 
\item \label{step:alice-cheat} \textbf{Test:} If $(x_0, x_1, x_2)=(0,0,0)$ and Bob has sent $(y_0, y_1,y_2)$ such that $y_a = 1$, then Alice aborts. 
\item Alice outputs $(x_0, x_1)$ and Bob outputs $(b, y_a)$. \label{step:output} 
\end{enumerate}
\end{protocol}

\begin{figure}[!h]
\centering
\begin{tikzpicture}
\node at (0,0) {\includegraphics[scale=0.5]{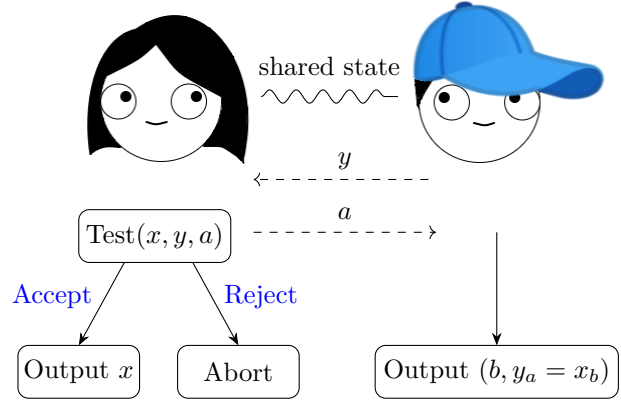}};
\draw[decorate, decoration = snake] (-1.1,0) -- (0.7,0);
\node at (-0.2, 0.4) {\small shared state};
\draw[->, dashed] (1.1,-1.1) -- (-1.2,-1.1);
\node at (0,-0.85) {\small $y$};
\draw[->, dashed] (-1.2, -1.8) -- (1.2,-1.8);
\node at (0, -1.55) {\small $a$};
\draw[rounded corners] (-3.5, -2.2) rectangle (-1.5, -1.5) node[pos=0.5] {\small Test$(x,y,a)$};
\draw[-Stealth] (-2.9,-2.2) -- (-3.5,-3.3);
\node at (-3.85,-2.65) {\small \color{blue} Accept};
\draw[rounded corners] (-4.3,-4) rectangle (-2.7,-3.3) node[pos=0.5] {\small Output $x$};
\draw [-Stealth] (-2,-2.2) -- (-1.4,-3.3);
\node at (-1.1,-2.65) {\small \color{blue} Reject};
\draw[rounded corners] (-2.2,-4) rectangle (-0.6,-3.3) node[pos=0.5] {\small Abort};
\draw[-Stealth] (2,-1.8) -- (2, -3.3);
\draw[rounded corners] (0.4, -4) rectangle (3.6, -3.3) node[pos=0.5] {\small Output $(b,y_a = x_b)$};
\end{tikzpicture}
\caption{Schematic depiction of the messages sent in Protocol~\ref{prot:XOTMS}.}
\end{figure}
   
Intuitively, the test step in Protocol~\ref{prot:XOTMS} serves as a weak test that the magic square winning condition is fulfilled (though the test only occurs with somewhat small probability). 
This provides a way to partially certify that Bob has measured his share of the state \emph{before} learning Alice's input. 
Note that if the test is not performed, Bob need not actually measure his state before receiving Alice's input. So even when the states are trusted, he has the potential to perfectly learn Alice's output after learning her input, by simply performing the same measurement as Alice on his state. 
However, forcing him to perform the magic square measurement ``deletes'' some of this information; a notion referred to as \emph{certified deletion} in~\cite{FM18}. 

We note that the test is asymmetric --- only Alice checks if Bob's outputs satisfy a certain condition. This is because if we were to consider the protocol without this test, \emph{only} Bob has any possibility of cheating: in that version Bob would not communicate at all with Alice, so she cannot learn his input (simply due to no-signalling). In our actual protocol with the test, security against Alice is no longer guaranteed via no-signalling. But as it turns out, we can still show Alice cannot perfectly learn Bob's input from his message (if Bob's measurement is trusted), so we do not add any further testing from Bob's side.

It is easily verified that Protocol~\ref{prot:XOTMS} is complete when Alice and Bob use a strategy that wins the magic square game with certainty. 
To prove its soundness, we compute the cheating probabilities using appropriate SDPs, as described in Section~\ref{sect:proofsDD}, for the trusted state and untrusted state cases. 
We summarize these cheating probabilities in
the following theorem, and below it we explain why Alice and Bob can cheat perfectly if they control the other's measurement devices. 

\newpage
\begin{theorem} 
\label{th:ms-xot} 
If Alice and Bob play the ``canonical'' strategy for the magic square game (see Section~\ref{sect:ms-intro}),  
then the cheating probabilities for Alice and Bob in Protocol~\ref{prot:XOTMS} are 
as shown
in the table below, rounded to 5 decimal places.
\begin{center}
\begin{tabular}{|c|c|c|}
\hline
\rule{0pt}{2.5ex} 
& $P^\xot_A$ & $P^\xot_B$ \\[0.5ex]
\hline
\rule{0pt}{2.5ex} Trusted state 
& 0.83333 
& 0.93628 
\\[0.5ex]
\hline
\rule{0pt}{2.5ex} Untrusted state 
& 0.87268 
& 0.94096 
\\[0.5ex]  
\hline
\rule{0pt}{2.5ex} Untrusted measurement 
& 1 
& 1 
\\[0.5ex]  
\hline
\end{tabular} 
\end{center}
\end{theorem} 

\paragraph{Untrusted-measurement scenario.} 
As the last row of the table in Theorem~\ref{th:ms-xot} indicates, Protocol~\ref{prot:XOTMS} is not fully device-independent. 
To see this, note that if Bob were to control Alice's measurements, he can force ${(x_0, x_1) = (0, 1)}$ to always occur (regardless of the state, by performing a trivial measurement), and then he will never be tested. 
This also fixes the value of $(x_0, x_1)$, so Bob knows it perfectly.
Conversely, if Alice controls Bob's measurement, she can force the output to be such that $y_b = 1$ and the other two bits are $0$. Then Bob's message fully reveals $b$ to Alice.

\subsection{A fully device-independent protocol}   
\label{sect:DI}

We now aim to find an XOT protocol based on the magic square game that is fully device-independent. We shall first clarify the premises and assumptions in such a setting. 
Specifically, we shall suppose that Alice and Bob each possess one of a pair of black boxes, each of which will accept a classical input in $\{0,1,2\}$ and return a classical output in $\{0,1\}^3$. We shall only require a single use of these boxes. 
In the honest scenario, the boxes will simply be implementing the ideal magic square states and measurements. 
If either party is dishonest, however, we shall suppose only that the boxes' behaviour can be modelled as follows: the boxes share some entangled state between them, and when one of the boxes receives an input, it returns the output of some measurement (conditioned on the input) performed on its share of the quantum state, \emph{without broadcasting either its input or output} to any party other than the one holding the box\footnote{To take a slightly different perspective (used in e.g.~\cite{ADF+18}), we could suppose that the honest party is able to ``shield'' their lab in a way such that signals cannot be broadcast out of it once they have supplied the input to their box.}. 
While the honest party can only interact with the box as specified, the dishonest party is able to ``open'' any box they possess and perform arbitrary quantum operations or measurements on the share of the state held by that box. 
However, the dishonest party cannot interact with or change the behaviour of a box while it is in the honest party's possession.

This describes the general device-independent setting. 
For the purposes of this work, we shall impose a small additional assumption on the states and measurements the boxes implement, namely that they are described by a tensor product of Hilbert spaces, one for each box.
More general scenarios could be considered (for instance, one could require only that the two boxes' measurements commute), but we do not consider them within this work --- we shall briefly discuss this point when outlining our proof techniques below.

It would seem difficult to design a secure protocol under such weak assumptions. 
However, we can exploit the fact that many nonlocal games (including the magic square game) exhibit the important property of \emph{self-testing} or \emph{rigidity}: if the boxes win the game with probability equal to the maximum quantum value, then they must be implementing the ideal state and measurements (up to trivial isometries). 
A robust version of this statement is formally expressed as Lemma~\ref{lemm:rigidity} in the next section.

This suggests the following idea to make Protocol~\ref{prot:XOTMS} fully device-independent: we introduce an initial step where with some probability, either party may ask to perform a single-shot test of whether the boxes win the magic square game. They do this by sending over a uniformly random input for the other party's box, and asking the other party to send over the output that they get by entering this input into their box (in the honest case). In reality a dishonest party may discard the input sent by the honest party and produce this claimed output by unknown means; but this process of obtaining the output can still be modelled by a measurement, to which we can apply a self-testing result.

To prevent a dishonest party from always calling for a test and then declaring a pass, we shall enforce that a party calling for a test must then cede all control if the test is passed, performing a trivial XOT protocol that is perfectly secure against them. Specifically, if Alice calls for the test and it passes, then Protocol \ref{prot:Bad1} is performed, and if Bob calls for the test and it passes, then Protocol \ref{prot:Bad2} is performed. 
(Here we are using several properties of Protocols \ref{prot:Bad1} and \ref{prot:Bad2}. First, they have perfect completeness, hence ensuring that the overall protocol we construct is perfectly complete, as is the focus of this work. Second, they achieve perfect security against one party in a ``trivial'' device-independent sense, because they do not need trusted quantum states or measurements, albeit at the cost of being fully insecure against the other party.)\footnote{One may want to implement protocols that are better than Protocols \ref{prot:Bad1} and \ref{prot:Bad2} in this construction to improve its performance (for instance, protocols that achieve partial security against the tested party in addition to the perfect security against the testing party). However, proving such improved performance would require device-independent security guarantees on whatever protocols are used in place of Protocols \ref{prot:Bad1} and \ref{prot:Bad2}, and we do not know any other device-independent protocols for XOT.}
If a test is not called, the parties simply perform Protocol~\ref{prot:XOTMS}. This idea of the parties probabilistically performing a self-test or a device-dependent protocol, and implementing a trivial protocol if the self-test passes, was also used in \cite{ACK+14}. We describe this construction formally as Protocol~\ref{prot:di-xot} below. 

Qualitatively, Protocol~\ref{prot:di-xot} imposes a ``tradeoff'' for the cheating party between passing the test (if it is called) and the extent to which they deviate from the ideal implementation of Protocol~\ref{prot:XOTMS}. 
More specifically, if (say) Bob is dishonest, he could cheat perfectly if Alice decides to test, by having both boxes implement the honest behaviour, and he could also cheat perfectly if Alice decides not to test, but to do so he needs to modify Alice's box's behaviour (since our device-dependent arguments show that perfect cheating is impossible when Alice's box is honest). 
Since Alice's box must behave differently in the two scenarios and Bob cannot change how that box behaves once the protocol starts, Alice can constrain his cheating probability by randomly choosing between testing and not testing. 
A similar argument applies to cheating Alice.

Note that for this reasoning to be valid, it is important that the honest party's box cannot be allowed to detect whether it is being subjected to a magic square test or whether it is being used for Protocol~\ref{prot:XOTMS} (we are implicitly assuming that the honest party's box behaves the same way in both situations). 
An assumption of this nature is typically required in device-independent protocols that involve performing a test with some probability, e.g.~\cite{ACK+14,ADF+18}. 
In particular, 
as observed in~\cite{ACK+14}, we note that
if the behaviour of the boxes could be time-dependent, then the honest party must ensure they provide the input to their box at a fixed pre-determined time, regardless of whether the box is being tested or used for Protocol~\ref{prot:XOTMS}. 
  
\begin{protocol}[Quantum XOT protocol from the magic square game \emph{with extra test steps}] \label{prot:di-xot}
\hfill
\begin{enumerate}
\item Alice flips a coin whose outcome is 0 with probability $1-q^A$, to obtain $c^A \in \{0,1\}$, which she sends to Bob.
\item
\begin{enumerate}
\item If $c^A = 0$, Bob flips a coin whose outcome is 0 with probability $1-q^B$, to obtain $c^B \in \{0,1\}$, which he sends to Alice.
\item If $c^A = 1$, Alice picks uniformly random $a^A, b^A\in \{0,1,2\}$, enters $a^A$ into her box, and sends $b^A$ to Bob.
\end{enumerate}
\item
\begin{enumerate}
\item If $c^A = 0, c^B = 0$, Alice and Bob perform Protocol~\ref{prot:XOTMS} henceforth.
\item If $c^A=0, c^B = 1$, Bob picks uniformly random $a^B, b^B \in \{0,1,2\}$, enters $b^B$ into his own box, and sends $a^B$ to Alice.
\item If $c^A = 1$,  Alice receives $y^A$ from Bob, has her output $x^A$ from her own box, and checks if they satisfy $x^A_{b^A} = y^A_{a^A}$. 
If not, she aborts.
\end{enumerate}
\item If $c^A=0, c^B=1$, Bob receives $x^B$ from Alice, has his own output $y^B$, and checks if they satisfy $x^B_{b^B} = y^B_{a^B}$. If not, he aborts.
\item
\begin{enumerate}
\item If $c^A = 1$ and Alice has not aborted, Alice and Bob perform Protocol~\ref{prot:Bad1} henceforth.
\item If $c^A = 0, c^B = 1$ and Bob has not aborted, Alice and Bob perform Protocol~\ref{prot:Bad2} henceforth.
\end{enumerate}
\end{enumerate}
\end{protocol}

\begin{figure}[!h]
\centering
\begin{tikzpicture}
\node at (0,0) {\includegraphics[scale=0.55]{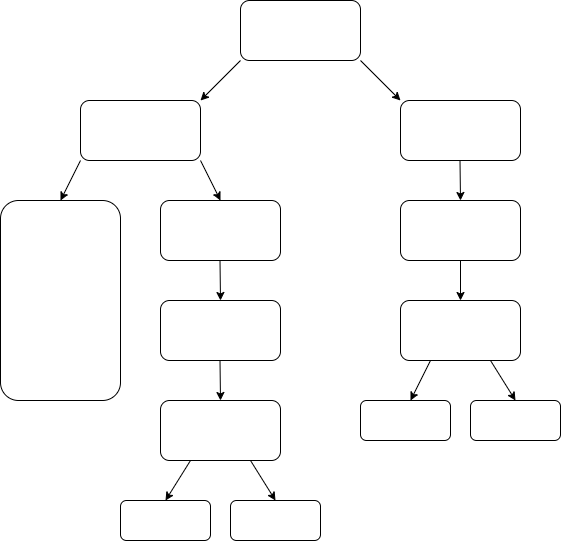}};
\node at (0.4,4.8) {\small Alice flips};
\node at (0.4, 4.4) {\small a coin};
\node at (-1.8,3.8) {\small \color{blue} $c^A = 0$};
\node at (2.7, 3.8) {\small \color{blue} $c^A=1$};

\node at (-2.7, 2.9) {\small Bob flips};
\node at (-2.7, 2.5) {\small a coin};
\node at (-4.75, 1.9) {\small \color{blue} $c^B=0$};
\node at (-0.65, 1.9) {\small \color{blue} $c^B=1$};
\node at (-4.2, -0.5) {\small Protocol~\ref{prot:XOTMS}};
\node at (-1.15, 1) {\small Bob sends};
\node at (-1.15, 0.6) {\small input $x^B$};
\node at (-1.15,-1) {\small Alice sends};
\node at (-1.15,-1.4) {\small $a^B$};
\node at (-1.15, -2.85) {\small Bob checks};
\node at (-1.15, -3.3) {\small if $x^B_{b^B} = y^B_{a^B}$};
\node at (-2.6, -4) {\small \color{blue} Accept};
\node at (-2.25, -4.9) {\small Protocol~\ref{prot:Bad2}};
\node at (0.2, -4) {\small \color{blue} Reject};
\node at (-0.1, -4.9) {\small Abort};

\node at (3.5, 2.9) {\small Alice sends};
\node at (3.5, 2.5) {\small input $y^A$};
\node at (3.5, 1) {\small Bob sends};
\node at (3.5, 0.6) {\small $b^A$};
\node at (3.5, -0.95) {\small Alice checks};
\node at (3.5, -1.4) {\small if $x^A_{b^A} = y^A_{a^A}$};
\node at (2.1, -2.1) {\small \color{blue} Accept};
\node at (2.45, -2.9) {\small Protocol~\ref{prot:Bad1}};
\node at (4.9, -2.1) {\small \color{blue} Reject};
\node at (4.6, -2.9) {\small Abort};
\end{tikzpicture}
\caption{Flowchart for Protocol~\ref{prot:di-xot}. 
}
\end{figure}
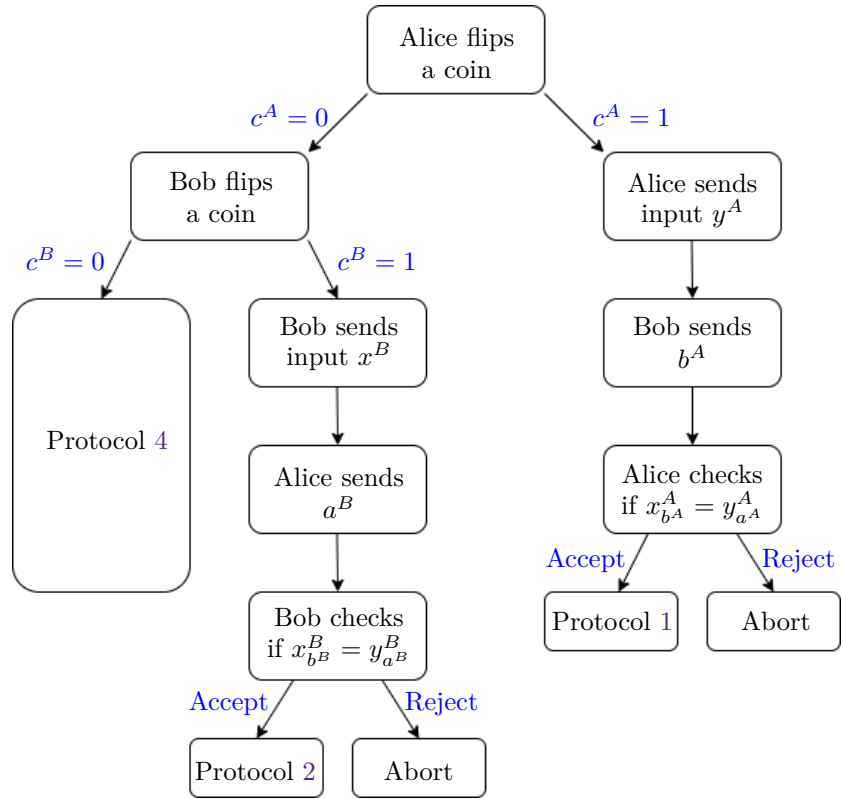

We give two soundness arguments for Protocol~\ref{prot:di-xot}. 
The first consists of explicit numerical bounds on the cheating probabilities, based on the family of SDPs known as the Navascu{\'e}s-Pironio-Ac{\'\i}n (NPA) hierarchy~\cite{NPA08}. 
We state the results as Theorem~\ref{th:ms-dixot} below, and give the proof in Section~\ref{sect:DINPA}. 
The second is an analytic proof
that the cheating probabilities are bounded away from $1$,
based on the robust self-testing bounds for the magic square game~\cite{WBMS16,CN16}. 
We state this result formally as Theorem~\ref{th:ms-dixot-2} below, and give the proof in Section~\ref{sect:continuity}.

Note that the NPA hierarchy works with commuting measurement operators rather than a tensor product structure, hence it might be possible that the first proof approach would work with the weaker assumption of commuting measurements instead. However, there is a technical issue that we begin our proofs with a reduction to projective measurements (see Appendix \ref{app:projs}), which we have not generalized to the commuting measurements case. 

\begin{theorem} 
\label{th:ms-dixot} 
Upper bounds on the cheating probabilities for  Protocol~\ref{prot:di-xot} (the fully device-independent scenario) with $q^A = 0.6$, $q^B = 0.6$ are given below, rounded to 5 decimal places.
\begin{table}[!h]
\centering
\begin{tabular}{|c|c|c|}
\hline
\rule{0pt}{2.5ex} Upper bounds & $\PA$ & $\PB$ \\[0.5ex]
\hline
\rule{0pt}{2.5ex} Fully Device-Independent 
& 0.96440 
& 0.99204 
\\[0.5ex]
\hline 
\end{tabular} 
\end{table} 
\end{theorem} 

\begin{remark}
The choices for $q^A$ and $q^B$ in Theorem~\ref{th:ms-dixot} were made by computing the bounds for different choices of $q^A$ and $q^B$ in intervals of $0.1$, then simply taking the value that yields the best bounds on the cheating probabilities.
We note that the result obtained for $\PB$ is rather close to $1$; however, the significant figures shown here are within the tolerance levels of the solver. 
\end{remark}

\begin{theorem}\label{th:ms-dixot-2}
For any $q^A,q^B > 0$ in Protocol~\ref{prot:di-xot}, there exists some $\delta > 0$ such that $P^\xot_A, P^\xot_B \leq 1 - \delta$. 
\end{theorem}  

\subsection{Our security analysis approach} 

To analyze the security of the protocols for which perfect cheating is not possible, we require a way to calculate or upper bound the optimal cheating probabilities. 
For most of the bounds in this paper, we use the formalism of semidefinite programming to calculate optimal cheating probabilities. In the device-dependent case, i.e., in the proofs of Theorems \ref{th:Oldprotocol} and \ref{th:ms-xot}, our approach is similar to the SDP-based analysis of coin-flipping in Kitaev~\cite{Kit03} and general protocols in Gutoski and Watrous~\cite{GW07}.  

A semidefinite program (SDP) is an optimization problem of the form 
\begin{align*} 
\text{maximize} :                & \quad \tr(AX) \\
\text{subject to} : & \quad \Phi(X) = B \\
                  & \quad X \in \Pos(\calX)  
\end{align*}  
where $\calX$ and $\calY$ are complex Euclidean spaces, $\Phi : \Herm(\calX) \to \Herm(\calY)$ is a linear, Hermiticity-preserving transformation, and $C \in \Herm(\calX)$ and $B \in \Herm(\calY)$ are Hermitian. 
When the spaces are clear from context, we write $X \succeq Y$ 
to mean
that $X - Y$ is (Hermitian) positive semidefinite, noting that the special case $X \succeq 0$ simply means $X$ is positive semidefinite. 
We use $X \succ Y$ to mean that $X - Y$ is positive definite.

The \emph{dual} of the above optimization problem (called the \emph{primal} in this context) is the problem
\begin{align*}
\text{minimize:} & \quad \tr(BY) \\
\text{subject to:} & \quad \Phi^*(Y) \succeq A \\
  & \quad Y \in \Herm(\calY).
\end{align*}
Under mild conditions, for example that there exists an $X \succ 0$ that is primal feasible and a $Y$ that is dual feasible, the optimal values of the primal and dual problems are equal.

In the device-independent case, the proof of Theorem \ref{th:ms-dixot} is via the NPA hierarchy, which is an infinite hierarchy of SDPs (see Appendix~\ref{NPA}). The SDP corresponding to any finite ``level'' of the NPA hierarchy gives an upper bound on the cheating probabilities in Theorem \ref{th:ms-dixot}. Note however that the NPA hierarchy is inherently an optimization over all possible states \emph{and} measurements, and thus is only suitable to be applied to device-independent protocols (where additional test steps are included to certify the state and/or measurements). Indeed, if we were to attempt to use the NPA hierarchy to bound the cheating probabilities of Protocols~\ref{prot:non-di-xot}~and~\ref{prot:XOTMS}, we would get the trivial value of $1$. This is why we set up specialized (and finite) SDPs for Theorems \ref{th:Oldprotocol} and \ref{th:ms-xot} instead.\footnote{Strictly speaking, although Protocol \ref{prot:XOTMS} is used as an intermediate protocol to construct Protocol \ref{prot:di-xot}, our proof of Theorem \ref{th:ms-dixot} via the NPA hierarchy does not require Theorem \ref{th:ms-xot}. However, we designed Protocol \ref{prot:di-xot} after making sure that the device-dependent component of it is secure, and thus Theorem \ref{th:ms-xot} is a useful intermediate step: the analytical proof of Theorem \ref{th:ms-dixot-2}, which does require Theorem \ref{th:ms-xot}, reflects this.} 
    
Finally, although our cheating probabilities are obtained by numerically solving SDPs, we note the following points:

\begin{itemize}
\item First, since each value is computed as an SDP, the dual SDPs yield certified upper bounds on the cheating probabilities. (The primal and dual values of all SDPs we solved are identical to more significant figures than displayed in our theorem statements.) Should there be concerns about floating-point errors, one could in principle find rational approximations to the dual solutions and thus obtain a bound certified by exact arithmetic, but this process is tedious and will not be included in this work.
\item Second, we provide the code used (see below) so the SDP implementation may be verified and/or applied for any protocol modifications of interest.
\item Third, the ``punchline'' of this work is that there exists a fully device-independent protocol where Alice and Bob's cheating probabilities are bounded away from $1$. 
We show this theoretically and also provide the numerical upper bounds to give an idea of what the cheating probabilities may look like. 
\end{itemize} 

\subsection*{Computational platform} 

The optimal value of most semidefinite programs can be approximated to within any constant precision efficiently. 
In particular, our computations were performed using the MATLAB packages QETLAB~\cite{qetlab} and YALMIP \cite{yalmip} with solver MOSEK~\cite{mosek}. Some of the calculations reported here were performed using the Euler cluster at ETH Z\"{u}rich. The code used to compute the SDPs will be maintained at 
\begin{equation*} 
\href{https://github.com/ernesttyz/dixot}{\textup{https://github.com/ernesttyz/dixot}}. 
\end{equation*}

\subsection{Paper organization} 

The next section details the magic square game and a rigidity result which is needed for this work. 
In Section~\ref{sect:proofsDD}, we use semidefinite programming to bound the cheating probabilities of Protocols~\ref{prot:non-di-xot}~and~\ref{prot:XOTMS}. 
The security of our device-independent protocol is analyzed numerically in Section~\ref{sect:DINPA}, and we refer the reader to Appendix~\ref{NPA} for background on the NPA hierarchy which we use in the analysis. 
We conclude in Section~\ref{sect:continuity} with the analytical bounds on the device-independent protocol. 

\section{The magic square game strategy and a rigidity lemma}
\label{sect:ms-intro} 
  
In this section, we discuss the optimal magic square game strategy and present a rigidity result required for this work.  
  
The optimal quantum strategy for magic square can be described as follows. Alice and Bob share the state
\begin{equation} \ket{\Psi^\ms} = \frac{1}{\sqrt{2}}(\ket{00}_{\calX_0\calY_0} + \ket{11}_{\calX_0\calY_0})\frac{1}{\sqrt{2}}(\ket{00}_{\calX_1\calY_1} + \ket{11}_{\calX_1\calY_1}) \end{equation}
with Alice holding the registers $\calX_0\calX_1$, and Bob holding the registers $\calY_0\calY_1$. The measurements of Alice and Bob are given in Table~\ref{tab:ms-meas}. 

\begin{table}[!h]
\centering
\begin{tabular}{|c|c|c|c|}
\hline
\diagbox{$a$}{$b$} & 0 & 1 & 2 \\\hline
\rule{0pt}{4.5ex}0 & $\begin{aligned} \Pi^0_{00} = \kb{0}\otimes\Id \\ \quad \end{aligned}$ & $\begin{aligned} \Pi^0_{01} = \Id\otimes\kb{0} \\ \quad \end{aligned}$ &  $\begin{aligned} \Pi^0_{02} & = \kb{0}\otimes\kb{0} \\ & \quad + \kb{1}\otimes\kb{1} \end{aligned}$ \\[0.5cm]
  & $\begin{aligned} \Pi^1_{00} = \kb{1}\otimes\Id \\ \quad \end{aligned}$ & $\begin{aligned} \Pi^1_{01} = \Id\otimes\kb{1} \\ \quad \end{aligned}$ & $\begin{aligned} \Pi^1_{02} & = \kb{0}\otimes\kb{1} \\ & \quad + \kb{1}\otimes\kb{0} \end{aligned}$ \\[0.5cm] \hline
\rule{0pt}{5ex}1 & {$\begin{aligned} \Pi^0_{10} = \Id\otimes \kb{+} \\ \quad \end{aligned}$} & {$\begin{aligned} \Pi^0_{11} = \kb{+}\otimes\Id \\ \quad \end{aligned}$} & {$\begin{aligned} \Pi^0_{12} & = \kb{+}\otimes\kb{+} \\ & \quad + \kb{-}\otimes\kb{-}\end{aligned}$} \\[0.5cm]
  &  $\begin{aligned} \Pi^1_{10} = \Id\otimes\kb{-} \\ \quad \end{aligned}$ & $\begin{aligned} \Pi^1_{11} =  \kb{-}\otimes\Id \\ \quad \end{aligned}$ & {$\begin{aligned}\Pi^1_{12} & = \kb{+}\otimes\kb{-} \\ & \quad + \kb{-}\otimes\kb{+}\end{aligned}$} \\[0.5cm]\hline
\rule{0pt}{5ex}2 & $\begin{aligned} \Pi^0_{20} & = \kb{1}\otimes\kb{+} \\ & \quad + \kb{0}\otimes\kb{-}\end{aligned}$ & $\begin{aligned} \Pi^0_{21} & = \kb{+}\otimes\kb{1} \\ & \quad + \kb{-}\otimes\kb{0}\end{aligned}$ & $\begin{aligned} \Pi^0_{22} & = \kb{+i}\otimes\kb{+i} \\ & \quad + \kb{-i}\otimes\kb{-i}\end{aligned}$ \\[0.5cm]
  & $\begin{aligned} \Pi^1_{20} & = \kb{0}\otimes\kb{+} \\ & \quad + \kb{1}\otimes\kb{-}\end{aligned}$ & $\begin{aligned} \Pi^1_{21} & = \kb{+}\otimes\kb{0} \\ & \quad + \kb{-}\otimes\kb{1}\end{aligned}$ & $\begin{aligned} \Pi^1_{22} & = \kb{+i}\otimes\kb{-i} \\ & \quad + \kb{-i}\otimes\kb{+i}\end{aligned}$ \\[0.5cm] \hline
\end{tabular}
\caption{Possible measurements for either party in the quantum strategy for magic square.}
\label{tab:ms-meas}
\end{table}

On input $a$, Alice simultaneously performs the three 2-outcome measurements $\{(\Pi^0_{ab}, \Pi^1_{ab})\}_b$ in the row indexed by $a$ in Table~\ref{tab:ms-meas} (it can be checked that the three measurements in every row are compatible, so they can be performed simultaneously) on her registers $\calX_0\calX_1$. Her output $(x_0, x_1, x_2)$ is the output of the three measurements (in order). Similarly, on input $b$, Bob simultaneously performs the three 2-outcome measurements 
$\{(\Pi^0_{ab}, \Pi^1_{ab})\}_a$
in the column indexed by $b$ 
(the three measurements in every column are also compatible)
on his registers $\calY_0\calY_1$, and gives the outcomes of the three measurements as his output $(y_0,y_1,y_2)$.

Clearly, the measurement Alice performs to output $x_b$ is the same as the measurement Bob performs to output $y_a$. Since these measurements are performed on maximally entangled states, one can show Alice and Bob always get the same outcome for $x_b$ and $y_a$. Also, it can be verified that these measurements always produce outcomes satisfying the parity conditions $x_0\oplus x_1\oplus x_2 = 0$ and $y_0\oplus y_1\oplus y_2 = 1$ (this holds regardless of the state). It can be shown that the output distribution is uniform over all combinations that win the magic square game, i.e.~$\Pr(xy|ab) = 1/8$ if $x_b=y_a$ (and $x,y$ satisfy the parity conditions), and $\Pr(xy|ab) = 0$ otherwise. 

The above description views Alice and Bob as performing 8-outcome measurements (via a sequence of three 2-outcome measurements). However, since for \emph{any} state the measurements always produce outputs satisfying the parity conditions, we can equivalently suppose Alice and Bob measure to determine only $(x_0,x_1)$ and $(y_0,y_1)$, with the last bit for each determined by the parity conditions. (This is consistent with the way we defined the magic square game earlier.) These are 4-outcome measurements that can be expressed in terms of the $\Pi$ operators from Table~\ref{tab:ms-meas} as
\begin{equation}
M^\ms_{x_0x_1|a} = \Pi^{x_0}_{a0}\Pi^{x_1}_{a1} \quad \quad N^\ms_{y_0y_1|b} = \Pi^{y_0}_{0b}\Pi^{y_1}_{1b}. \label{eq:honPVM}
\end{equation}
It can be checked that each $M^\ms_{x_0x_1|a}$ and $N^\ms_{y_0y_1|b}$ is a rank-1 projector. Since the measurements for $x_0$ and $x_1$ (resp.~$y_0$ and $y_1$) commute for every $a$ (resp.~$b$), the product of the $\Pi$ operators in each case can be taken in either order.

Certain nonlocal games exhibit the property that the quantum strategies achieving their optimal values are essentially unique. That is, if a quantum strategy wins with probability within $\eps$ of the optimal probability, that strategy must be $\delta(\eps)$-close to the ideal strategy for the game, up to certain local operations. 
This property of rigidity or self-testing was described first in~\cite{MY98}, and has been shown for a variety of nonlocal games since.

\cite{WBMS16} originally gave a proof of the rigidity of a version of the magic square game which is slightly different from ours. \cite{CN16} showed that the rigidity statement also holds for the version of the magic square game we use. However, both of these results show the self-testing of some operators that are related to Alice and Bob's measurement operators in the magic square game, but not the measurement operators themselves. It is not immediately clear how to self-test the measurement operators themselves from their results. In Appendix~\ref{sect:rigid-proof} we derive the following lemma for self-testing of the measurement operators of the magic square strategy.

\begin{lemma}\label{lemm:rigidity}
Consider any state $\ket{\rho}$ on registers $\calX\calY$ and projective measurements $M_{x|a},N_{y|b}$ such that $M_{x|a}$ acts only on $\calX$ and $N_{y|b}$ acts only on $\calY$. If this state and measurements win the magic square game with probability $1-\eps$, then there exist local isometries $V_A: \calX \to \calX_0\calX_1\calJ_A$ and $V_B: \calY \to \calY_0\calY_1\calJ_B$ and a state $\ket{\jk}$ on $\calJ_A\calJ_B$ such that for all $a, b, x, y$, we have 
\begin{gather*}
\norm{(V_A \otimes V_B) \ket{\rho} - \ket{\Psi^\ms}\otimes\ket{\jk}}_2 \leq O(\eps^{1/4}),\\
\norm{(V_A \otimes V_B) (M_{x|a} \otimes \Id) \ket{\rho} - ((M^\ms_{x|a} \otimes \Id) \ket{\Psi^\ms})\otimes\ket{\jk}}_2 \leq O(\eps^{1/4}),\\
\norm{(V_A \otimes V_B) (\Id \otimes N_{y|b}) \ket{\rho} - ((\Id \otimes N^\ms_{y|b}) \ket{\Psi^\ms})\otimes\ket{\jk}}_2 \leq O(\eps^{1/4}),
\end{gather*}
where $\ket{\Psi^\ms}$, $M^\ms_{x|a}$, and $N^\ms_{y|b}$ denote the ideal state and measurements in the magic square game, 
as described above.
\end{lemma} 

\section{On the security of Protocol~\ref{prot:non-di-xot} and Protocol~\ref{prot:XOTMS}} 
\label{sect:proofsDD}

In this section we formulate the cheating probabilities as SDPs and solve them numerically to get their values.

\subsection{Security analysis of Protocol~\ref{prot:non-di-xot}} 
\label{sect:old-proof}

In this section, we prove that when the devices are trusted in Protocol~\ref{prot:non-di-xot}, then $\PA = 1/2$ and $\PB = 3/4$. 
Before continuing, recall that Protocol~\ref{prot:non-di-xot} is an adaptation of the protocol in~\cite{CKS13} where Alice's actions are the exact same and so are the intentions of a dishonest Bob. 
Therefore, we can import $\PB = 3/4$ directly from the security analysis of that protocol. 

The rest of this subsection is devoted to proving $\PA = 1/2$. 
Since Bob never aborts, all we need to ascertain is the ability for Alice to learn $b$ from her information contained in the first message. 
To do this, she must infer $b$ from the ensemble
\begin{equation} 
\left\{ \left( \frac{1}{3}, \rho_b = \tr_{\calY} (\kb{\psi^+_b} \right) : b \in \{ 0, 1, 2 \} \right\}. 
\end{equation} 
This is known as the \emph{quantum state discrimination problem}, and the optimal guessing probability can be written as the following SDP:
\begin{center}
\begin{minipage}{2in}
\centerline{\underline{Primal problem}}\vspace{-5mm}
\begin{align*}
\text{maximize:}   & \quad \frac{1}{3} \sum_{b \in \{0,1,2\}} \tr({E_b}{\rho_b}) \\
\text{subject to:} & \quad \sum_{b \in \{0,1,2\}} E_b = \Id_\calX \\
                   & \quad \forall b \, E_b \in \Pos(\calX)
\end{align*}
\end{minipage}
  \hspace*{12mm}
\begin{minipage}{2in}
\centerline{\underline{Dual problem}}\vspace{-5mm}
\begin{align*}
\text{minimize:}               & \quad \tr(\sigma) \\
\text{subject to:} & \quad \forall b \, \, \frac{1}{3} \rho_b \preceq \sigma \\
                   & \quad \sigma \in \Herm(\calX).
\end{align*}
\end{minipage}
\end{center} 
Note that the success probability can be seen to be 
given by this SDP
as it is a maximization over POVMs and the objective function is the probability of that POVM measurement producing the correct guess for $b$. 

Now, if Alice uses the POVM 
\begin{equation} 
\{E_0, E_1, E_2\} = \{\kb{0} + \kb{2}, \kb{1}, 0\}, 
\end{equation} 
we can see that 
\begin{equation} 
\PA 
\geq \frac{1}{3} \sum_{b \in \{0,1,2\}} \tr(E_b \rho_b)
= \frac{1}{3} + \frac{1}{3} \cdot \frac{1}{2} + 0 = \frac{1}{2}. 
\end{equation} 
Effectively what this measurement does is measure in the computational basis, then assign the outcomes $\ket{0}$ and $\ket{2}$ to the guess $b = 0$ and the outcome $\ket{1}$ to the guess $b = 1$. 
Note that $b = 2$ is never guessed in this strategy. 
All that remains to show is that $\PA \leq 1/2$. 
For this, we use the dual problem. 
Consider the dual feasible solution 
\begin{equation} 
\sigma = \frac{1}{6} \Id_{\calX}. 
\end{equation} 
It can be checked that $\sigma$ satisfies the dual constraints, i.e., $\sigma \succeq \frac{1}{3} \rho_b$ for all $b \in \{ 0, 1, 2 \}$. 
Since $\tr(\sigma) = 1/2$, we have that $\PA \leq 1/2$ by weak duality.  

\subsection{Security analysis of Protocol~\ref{prot:XOTMS}} 

We split this section up into the four cases, depending on whether or not the state is trusted at the beginning, and each cheating party. The cheating probabilities in Theorem \ref{th:ms-xot} are obtained by solving the SDPs obtained in each case. 

\paragraph{Cheating strategy formulation set-up.} 
Recall that Alice and Bob's ideal measurement operators in the magic square game are denoted by $M^\ms_{x_0x_1|a}$ and $N^\ms_{y_0y_1|b}$. The measurement corresponding to $a$ may be performed by applying a unitary that takes the $\{M^\ms_{x_0x_1|a}\}_{x_0x_1}$ basis to the computational basis, and then measuring in the computational basis. We denote this unitary by $U_a$ for each input $a$ to Alice. Similarly, the measurement corresponding to $b$ can be performed by applying $U_b$ and then measuring in the computational basis.

\subsubsection{Untrusted state, cheating Bob}    
  
In Protocol~\ref{prot:XOTMS}, Bob cheats successfully if he is able to output $(x_0, x_1)$ at the end of the protocol and Alice does not abort in step~\ref{step:alice-cheat}. 
The steps of Protocol~\ref{prot:XOTMS} that are relevant for determining whether Bob cheats successfully can then be described as follows:
\begin{enumerate}
\item Alice holds a state on two qubit registers $\calX_0\calX_1$.
\item \label{step:rho1} Bob sends two qubit registers $\calY_0\calY_1$ to Alice (encoding his claimed outcomes $(y_0,y_1)$). 
\item Alice prepares the state
\begin{equation} 
\ket{\phi_3}_{\calA\calA'} = \frac{1}{\sqrt{3}}(\ket{00}_{\calA\calA'} + \ket{11}_{\calA\calA'} + \ket{22}_{\calA\calA'}),
\label{eq:coherent}
\end{equation}
and sends $\calA'$ to Bob.
\item \label{step:rho2} Bob produces a guess $g_B$ for $(x_0,x_1)$
in the two-qubit register $\calG_B$. 
\item \label{step:cheatprob} The event of Bob cheating successfully is now determined by a joint measurement on the registers $\calA\calX_0\calY_0\calX_1\calY_1\calG_B$ (see Eq.~\eqref{eq:cheatmeas} below).
\end{enumerate}
We remark that in the actual protocol, the registers $\calA\calA'\calY_0\calY_1\calG_B$ are classical, but in the above description we have treated them as quantum. The cheating probability is unaffected by this change, because of the following observation: in principle, the classicality condition could be enforced by applying a pinching channel in the computational basis immediately after each of those registers is generated. Since no further operations are performed on these registers before the final measurement in step~\ref{step:cheatprob},\footnote{Strictly speaking, this is not true for $\calA'$, which Bob could operate on. Fortunately, the classicality condition for $\calA'$ can be enforced by applying a pinching channel to $\calA$ alone in the state~\eqref{eq:coherent}, without acting on $\calA'$. This pinching channel on $\calA$ can indeed be deferred as desired.}, we can defer these pinching channels to the point just before that measurement is performed. However, that measurement (Eq.~\eqref{eq:cheatmeas}) is diagonal in the computational basis of the registers $\calA\calY_0\calY_1\calG_B$, thus its outcome probabilities remain invariant regardless of whether those pinching channels were explicitly implemented beforehand. Hence those channels can be omitted entirely without changing the cheating probability.

Let $\rho_1$ and $\rho_2$ be Alice's states at the end of steps~\ref{step:rho1} and~\ref{step:rho2} respectively in the above description. These states are on registers $\calX_0\calY_0\calX_1\calY_1$ and $\calA\calX_0\calY_0\calX_1\calY_1\calG_B$ respectively. Since Alice does not perform any operations on the registers she holds between these steps, we must have
\begin{align}
\tr_{\calG_B}(\rho_2) = \tr_{\calA'}\left(\rho_1\otimes \kb{\phi_3}\right) = \rho_1 \otimes \frac{\Id_\calA}{3}.
\end{align}
The measurtement in step~\ref{step:cheatprob} can be described as follows. First, Alice performs a controlled unitary on registers $\calA\calX_0\calX_1$ given by
\begin{equation} 
U_{\calA\calX_0\calX_1} = \sum_{a \in \{0,1,2\}}\kb{a}_\calA\otimes (U_a)_{\calX_0\calX_1}.
\end{equation}
A computational-basis measurement is then performed on $\calA\calX_0\calY_0\calX_1\calY_1\calG_B$ to determine whether Bob cheats successfully. 
Letting $P_{\calY_0\calY_1\calA}$ denote the projector
\begin{align}
P_{\calY_0\calY_1\calA} & = \kb{10}_{\calY_0\calY_1}\otimes\kb{0}_{\calA} + \kb{01}_{\calY_0\calY_1}\otimes\kb{1}_{\calA} \nonumber \\
& \quad + \kb{00}_{\calY_0\calY_1}\otimes\kb{2}_{\calA} + \kb{11}_{\calY_0\calY_1}\otimes\Id_{\calA},
\end{align}
the projector  
corresponding to Bob cheating successfully is $\Id-\povm_{\rej}$, where
\begin{align}
\povm_{\rej} & = \kb{00}_{\calX_0\calX_1}\otimes \left(P_{\calY_0\calY_1\calA}\otimes\Id_{\calG_B} + (\Id_{\calY_0\calY_1\calA_1} - P_{\calY_0\calY_1\calA})\otimes\sum_{(x'_0,x'_1) \neq (0,0)}\kb{x'_0x'_1}_{\calG_B}\right) \nonumber \\
& \quad + \sum_{\substack{(x_0, x_1) \neq (0,0) \\ (x_0',x_1') \neq (x_0,x_1)}}\kb{x_0x_1}_{\calX_0\calX_1}\otimes\Id_{\calY_0\calY_1\calA}\otimes\kb{x_0'x_1'}_{\calG_B}.
\end{align}
The first term in $\povm_\rej$ corresponds to Alice's outputs being $x_0x_1=00$, in which case Bob may fail if he either fails to pass Alice's test correctly, or if his guess $x_0'x_1'$ is not equal to $x_0x_1$. The second term corresponds to Alice's output being anything other than $00$, in which case Alice does not conduct a test, and Bob only fails if he fails to guess $x_0x_1$ correctly. The projector $P_{\calY_0\calY_1\calA}$ has a term for all the cases in which Bob fails Alice's test when $x_0x_1=00$. For example, if $y_0y_1=01$ (and correspondingly $y_2=1$), then the test fails if $a=0$ in the $\calA$ register; if $y_0y_1=11$ (and correspondingly $y_2=1$), then the test always fails, since none of them can be equal to any $x_a$. $\Id_{\calY_0\calY_1\calA}-P_{\calY_0\calY_1\calA}$ then corresponds to Bob not failing the test, and this case he still fails overall when $x'_0x'_1\neq 00$, which is reflected in the operator on $\calG_B$ that $\Id_{\calY_0\calY_1\calA}-P_{\calY_0\calY_1\calA}$ is tensored with.

The cheating probability is thus given by $\tr((\Id - \povm_{\rej})U_{\calA\calX_0\calX_1}\rho_2U^\dagger_{\calA\calX_0\calX_1})$. To describe it using a single measurement, we can write this as $\tr(\povm'_\acc \rho_2)$, where 
\begin{equation}
\povm'_\acc = U^\dagger_{\calA\calX_0\calX_1}(\Id - \povm_{\rej})U_{\calA\calX_0\calX_1}. 
\label{eq:cheatmeas}
\end{equation}
Hence we conclude Bob's maximum cheating probability is bounded 
above by the following SDP:
\begin{center}
\begin{minipage}{2in}
\centerline{\underline{Primal problem}}\vspace{-5mm}
\begin{align*}
\text{maximize:}   & \quad  \tr(\povm'_\acc\rho_2) \\
\text{subject to:} & \quad \tr_{{\calG}_B}(\rho_2) = \rho_1\otimes \frac{\Id_\calA}{3} \\
                   & \quad \tr(\rho_1) = 1 \\
                   & \quad \rho_1 \in \Pos(\calX_0\calY_0\calX_1\calY_1) \\
                   & \quad \rho_2 \in \Pos(\calA\calX_0\calY_0\calX_1\calY_1\calG_B)
\end{align*}
\end{minipage}
  \hspace*{12mm}
\begin{minipage}{2in}
\centerline{\underline{Dual problem}}\vspace{-5mm}
\begin{align*}
\text{minimize:}   & \quad \eta \\
\text{subject to:} & \quad \frac{1}{3}\tr_\calA(Q) \preceq \eta\Id_{\calX_0\calY_0\calX_1\calY_1} \\
                   & \quad \povm'_\acc \preceq Q\otimes \Id_{{\calG}_B} \\
                   & \quad \eta \in \R \\
                   & \quad Q \in \Herm(\calA\calX_0\calY_0\calX_1\calY_1).
\end{align*}
\end{minipage}
\end{center}
For the purposes of numerically solving the primal, one can enforce that the registers $\calA\calY_0\calY_1\calG_B$ are classical (by applying pinching channels in the computational basis), since the objective function and constraints are unaffected by 
applying such a channel to $\rho_1,\rho_2$.
This allows dramatically faster solving of the SDP, since it eliminates many variables. In addition, the optimization variables can be assumed to have real-valued entries, since the constraints and objective function are described by real-valued coefficients. 

Solving the SDP numerically yields a value of $0.94096$. 

To show that this bound is tight, we shall argue that every feasible solution to the (primal) SDP corresponds to an attack Bob can physically perform. Consider any feasible states $\rho_1,\rho_2$ in the SDP, and let $\ket{\rho_1}_{\calX_0\calY_0\calX_1\calY_1\calR}, \ket{\rho_2}_{\calA\calX_0\calY_0\calX_1\calY_1\calG_B\calR'}$ be any purifications of these states onto some registers $\calR,\calR'$ respectively. The main constraint 
in the SDP implies that
\begin{align}
\tr_{\calG_B\calR'}(\kb{\rho_2}) = \tr_{\calA'\calR}\left(\kb{\rho_1}\otimes \kb{\phi_3}\right).
\end{align}
By Uhlmann's theorem, this implies (expanding the dimension of $\calR'$ if necessary) there exists an isometry $V:\calA'\calR \to \calG_B\calR'$ that sends $\ket{\rho_1}_{\calX_0\calY_0\calX_1\calY_1\calR} \otimes \ket{\phi_3}_{\calA \calA'}$ to $\ket{\rho_2}_{\calA\calX_0\calY_0\calX_1\calY_1\calG_B\calR'}$. Bob can hence implement the following attack: in the initial steps, Bob prepares the state $\ket{\rho_1}_{\calX_0\calY_0\calX_1\calY_1\calR}$ and keeps the purification $\calR$. Then after receiving $\calA'$ from Alice, he implements the isometry $V$ on $\calA'\calR$ to produce $\calG_B\calR'$. Discarding the purification $\calR'$, this produces the feasible state $\rho_2$ in the SDP, achieving the desired cheating probability. (To be precise, one should again incorporate appropriate pinching channels to fulfill the classicality requirements for a physical attack, but as argued above, these channels can indeed be implemented without changing the cheating probability.)

\subsubsection{Trusted state, cheating Bob}    
  
The fact that the state is trusted implies an extra constraint on $\rho_1$, namely $\tr_{\calY_0\calY_1}(\rho_1) = \frac{\Id_{\calX_0\calX_1}}{4}$. Note that this automatically takes care of the $\tr(\rho_1) = 1$ constraint. The rest of the argument applies in exactly the same way, so Bob's cheating probability in this case is bounded by the following SDP:
\begin{center}
\begin{minipage}{2in}
\centerline{\underline{Primal problem}}\vspace{-5mm}
\begin{align*}
\text{maximize:}   & \quad  \tr(\povm'_\acc\rho_2) \\
\text{subject to:} & \quad \tr_{{\calG}_B}(\rho_2) = \rho_1\otimes \frac{\Id_\calA}{3} \\
                   & \quad \tr_{\calY_0\calY_1}(\rho_1) = \frac{\Id_{\calX_0\calX_1}}{4} \\
                   & \quad \rho_1 \in \Pos(\calX_0\calY_0\calX_1\calY_1) \\
                   & \quad \rho_2 \in \Pos(\calA\calX_0\calY_0\calX_1\calY_1\calG_B)
\end{align*}
\end{minipage}
  \hspace*{12mm}
\begin{minipage}{2in}
\centerline{\underline{Dual problem}}\vspace{-5mm}
\begin{align*}
\text{minimize:}   & \quad \frac{1}{4} \tr\left(Q_2\right) \\
\text{subject to:} & \quad \frac{1}{3}\tr_\calA(Q_1) \preceq Q_2\otimes\Id_{\calY_0\calY_1} \\
                   & \quad \povm'_\acc \preceq Q_1\otimes \Id_{{\calG}_B} \\
                   & \quad Q_2 \in \Herm(\calX_0\calX_1) \\
                   & \quad Q_1 \in \Herm(\calA\calX_0\calY_0\calX_1\calY_1).
\end{align*}
\end{minipage}
\end{center} 

Solving this SDP numerically yields a value of $0.93628$. This value is a bit smaller than the case when the state is untrusted, as should be expected.

To see that this bound is tight, we follow the same argument as before, together with the fact that the new constraint completely captures the fact that the whole initial state is trusted (not just Alice's side). This is because the initial state is pure, and Bob is allowed to apply arbitrary operations on his side. This lets him produce any extension of Alice's side of the state, hence there is no loss of tightness by only imposing a constraint on Alice's side.

\subsubsection{Untrusted state, cheating Alice} 

\newcommand{\calYout}{\calZ} 
\newcommand{\calYp}{{\calY'}} 
In this section, for brevity we will use $\calY$ to denote the registers $\calY_0\calY_1$ together, and label it using a single index $y=(y_0,y_1)$. 
Alice sends some state in register $\calY$ (possibly while retaining a purification of the state in a register $\calYp$) to Bob, who
carries out the measurement $\{N^\ms_{y|b}\}_{y}$ on it, and sends the outcome $y$ back to Alice in the computational basis of a register $\calYout$.\footnote{In previous sections, we used the registers $\calY_0\calY_1$ themselves to encode the outcome, but the subsequent proof is easier to follow with a new register storing the outcome.} So honest Bob performs one of the three channels 
\begin{equation}
\Phi_b(\sigma) = \sum_{y}\tr(N^\ms_{y|b}\sigma)\kb{y}_{\calYout},
\label{eq:channsBob}
\end{equation}
each with probability $\frac{1}{3}$, and cheating Alice wants to learn which channel he has performed. 
In order to do this, Alice in general can start with a pure state $\kb{\nu}$ in $\calY\calYp$ of which she sends register $\calY$ to Bob, so the overall channel being performed on $\kb{\nu}$ is $\Phi_b\otimes\Id_{\calYp}$. Note that by considering the Schmidt decomposition, we can always assume $\calYp$ has the same dimension as $\calY$ (as long as we allow Alice to perform general POVMs later). After Bob has performed the channel, Alice can carry out a POVM measurement $\{\povm_b\}$ on $\calYout \calYp$ to find out which channel was performed, with the POVM operator $\povm_b$ corresponding to Alice guessing channel $b$. Conditioned on channel $\Phi_b$ being performed, the probability of Alice guessing $b$ (i.e.~being correct) is $\tr(\povm_b (\Phi_b\otimes \Id_{\calYp} (\kb{\nu}) ))$. Hence Alice wants to maximize over pure states $\kb{\nu}$ in $\calY \calYp$ and three-outcome measurements $\{\povm_b\}$ the following quantity:
\begin{equation}\label{alice_opt1}
\frac{1}{3}\sum_{b \in \{0,1,2\}} \tr\left(\povm_b (\Phi_b\otimes \Id_{\calYp} (\kb{\nu}) ) \right)
\end{equation}
and the maximal value thereof is her maximum cheating probability.

We claim that the maximum cheating probability is given by the optimal value of the following primal-dual SDP pair, which is a slight generalization of the SDP for distance between channels due to \cite{Wat09, GW07}:
\begin{center}
\begin{minipage}{2in}
\centerline{\underline{Primal problem}}\vspace{-5mm}
\begin{align*}
\text{maximize:}   & \quad \frac{1}{3} \sum_{b \in \{0,1,2\}} \tr(Q_b J(\Phi_b)) \\
\text{subject to:} & \quad \sum_{b \in \{0,1,2\}} Q_b = \Id_{\calZ}\otimes\sigma \\
                   & \quad \tr(\sigma) = 1 \\
                   & \quad \sigma \in \Pos(\calYp), \forall b \, Q_b \in \Pos(\calYout\calYp)
\end{align*}
\end{minipage}
  \hspace*{12mm}
\begin{minipage}{2in}
\centerline{\underline{Dual problem}}\vspace{-5mm}
\begin{align*}
\text{minimize:}   & \quad \lambda \\
\text{subject to:} & \quad \tr_{\calYout}(R) \preceq \lambda\Id_{\calYp} \\
                   & \quad \forall b \, \frac{1}{3}J(\Phi_b) \preceq R \\
                   & \quad \lambda \in \R, R \in \Herm(\calYout\calYp),
\end{align*}
\end{minipage}
\end{center}
where $J(\Phi_b)$ is the Choi-Jamio{\l}kowski matrix for $\Phi_b$, which is in $\Herm(\calYout\calYp)$ and given by 
\begin{equation} 
J(\Phi_b) = \Phi_b\otimes\Id_{\calYp} \left(\kb{\Omega}\right) = \sum_{y}\kb{y}_{\calYout}\otimes N^\ms_{y|b},
\end{equation}
where $\ket{\Omega} = \sum_{y} \ket{yy}_{\calY\calYp}$ is the (unnormalized) maximally entangled state between registers $\calY\calYp$, with respect to the computational basis of $\calY$.

To see that the optimization for Alice's discrimination probability is equivalent to the above SDP, note that Alice's initial state $\ket{\nu}_{\calY\calYp}$ can always be written in the form
\begin{equation} \ket{\nu}_{\calY\calYp} = \sum_{y,y'} \nu_{y y'}\ket{y}_\calY\ket{y'}_{\calYp}, \end{equation}
for some coefficients $\nu_{yy'}$. Define the matrix $K_\nu \in \rL(\calYp)$ as 
\begin{equation} 
K_\nu = \sum_{y,y'}\nu_{yy'}\ketbra{y'}{y}.
\label{eq:op-vec}
\end{equation}
(This is the operator-vector isomorphism, up to a transpose.)
It is easily verified by direct calculation that $\ket{\nu} = (\Id_\calY\otimes K_\nu) \ket{\Omega}$, which implies
\begin{equation} \Phi_b\otimes\Id_{\calYp}(\kb{\nu}) 
= (\Id_{\calYout}\otimes K_\nu)J(\Phi_b)(\Id_{\calYout}\otimes K_\nu^\dagger).\end{equation}
Therefore, the quantity Alice wants to maximize (Eq.~\eqref{alice_opt1}) can be written as
\begin{equation} \sum_b \tr(\povm_b (\Phi_b\otimes \Id_{\calYp}(\kb{\nu}))) = \sum_b \tr\left((\Id_{\calYout} \otimes K_\nu^\dagger) \povm_b(\Id_{\calYout} \otimes K_\nu)J(\Phi_b)\right). \end{equation}
Furthermore, note that Eq.~\eqref{eq:op-vec} describes a bijective mapping between normalized initial states $\ket{\nu}$ and operators $K_\nu\in \rL(\calYp)$ satisfying $\Vert K_\nu\Vert_2 = 1$. Hence Alice's optimal discrimination probability can be written as
\begin{equation}\sup_{\{Q'_b\}_b \in \calS'}\frac{1}{3}\sum_b \tr(Q'_b J(\Phi_b)),\end{equation}
where $\calS'$ is the set
\begin{equation} \calS' = \left\{\{(\Id_{\calYout}\otimes K_\nu^\dagger)\povm_b(\Id_{\calYout}\otimes K_\nu)\}_{b \in \{0,1,2\}} : \povm_b \in \Pos(\calYout \calYp), \sum_b \povm_b = \Id_{\calYout \calYp}, K_\nu \in \rL(\calYp), \Vert K_\nu\Vert_2 = 1\right\}.
\label{eq:setSp}
\end{equation}
In comparison, the primal problem stated above is
\begin{equation} \sup_{\{Q_b\}_b \in \calS}\frac{1}{3}\sum_b \tr(Q_b J(\Phi_b))\end{equation}
where $\calS$ is the set
\begin{equation} \calS = \left\{ \{Q_b\}_{b \in \{0,1,2\}} : Q_b \in \Pos(\calYout\calYp), \sum_b Q_b = \Id_{\calYout}\otimes \sigma, \sigma \in \Pos(\calYp), \tr(\sigma) = 1 \right\}.
\label{eq:setS}
\end{equation}

We shall show that $\calS'$ and $\calS$ are the same set (thus the optimizations are the same), by proving that every element of the former lies in the latter, and vice versa. First, every element of $\calS'$ is a tuple $\{(\Id_{\calYout}\otimes K_\nu^\dagger)\povm_b(\Id_{\calYout}\otimes K_\nu)\}_{b \in \{0,1,2\}}$ defined by some $E_b$ and $K_\nu$ satisfying the conditions in~\eqref{eq:setSp}. It is easily verified that this tuple satisfies the conditions in~\eqref{eq:setS}, with
\begin{equation} 
Q_b = (\Id_{\calYout}\otimes K_\nu^\dagger)\povm_b(\Id_{\calYout}\otimes K_\nu),
\qquad
\sigma = K_\nu^\dagger K_\nu, 
\end{equation}
and hence it lies in $\calS$. As for the other direction, every element of $\calS$ is a tuple $\{Q_b\}_{b \in \{0,1,2\}}$ defined by some $Q_b$ and $\sigma$ satisfying the conditions in~\eqref{eq:setS}. It can be verified that this tuple is of a form satisfying the conditions in~\eqref{eq:setSp}, with
\begin{equation} 
\povm_b = (\Id_{\calYout}\otimes \sigma^{-1/2})Q_b(\Id_{\calYout}\otimes\sigma^{-1/2}) + \frac{\Id_{\calYout} \otimes \Pi^{\perp}_{\sigma}}{3} ,
\qquad 
K_\nu = \sigma^{1/2},
\end{equation}
where $\sigma^{-1/2}$ denotes the pseudoinverse of $\sigma^{1/2}$, and $\Pi^{\perp}_{\sigma}$ denotes the projector onto the kernel of $\sigma$. Hence this tuple lies in $\calS'$, as desired.

Solving the SDP numerically yields a value of $0.87268$.  

This bound is tight because the cheating probability is given by exactly the channel discrimination task we described, and the SDP (by construction) yields exactly the maximum discrimination probability.

\subsubsection{Trusted state, cheating Alice} 

Since the initial state is trusted, this means Alice initially holds the $\calX_0\calX_1$ registers of the state $\ket{\Psi^\ms}$, while Bob performs the measurement corresponding to input $b$ on his registers $\calY_0 \calY_1$, then receives an outcome which he sends to Alice by encoding it in the computational basis of a register $\calYout$. The state Alice then holds is
\begin{equation}
\rho_b = \sum_{y_0 y_1} \tr_{\calY_0\calY_1}((\Id \otimes N^\ms_{y_0y_1|b})\kb{\Psi^\ms}(\Id \otimes N^\ms_{y_0y_1|b})) \otimes \kb{y_0 y_1}_\calYout.
\end{equation}
Alice has to distinguish these states (for the three possible values of $b$), and hence her cheating probability is given by the state discrimination SDP from Section~\ref{sect:old-proof}:
\begin{center}
\begin{minipage}{2in}
\centerline{\underline{Primal problem}}\vspace{-5mm}
\begin{align*}
\text{maximize:}   & \quad \frac{1}{3} \sum_{b \in \{0,1,2\}} \tr(\povm_b \rho_b) \\
\text{subject to:} & \quad \sum_{b \in \{0,1,2\}} \povm_b = \Id_{\calX_0\calX_1\calYout} \\
                   & \quad \forall b \, \povm_b \in \Pos(\calX_0\calX_1\calYout)
\end{align*}
\end{minipage}
  \hspace*{12mm}
\begin{minipage}{2in}
\centerline{\underline{Dual problem}}\vspace{-5mm}
\begin{align*}
\text{minimize:}   & \quad \tr(\sigma) \\
\text{subject to:} & \quad \forall b\,\, \frac{1}{3} \rho_b \preceq \sigma \\
                   & \quad \sigma \in \Pos(\calX_0\calX_1\calYout).
\end{align*}
\end{minipage}
\end{center}

Solving the SDP numerically yields a value of $0.83333$. 
This value is a bit smaller than the case when the state is untrusted, as should be expected.

Again, this bound is tight because the cheating probability is given by exactly the state discrimination task we described.

\section{Numerically bounding the cheating probabilities for Protocol~\ref{prot:di-xot}}
\label{sect:DINPA}

In this section, we obtain expressions for Alice and Bob's cheating probabilities in Protocol \ref{prot:di-xot}, which can be optimized using the NPA hierarchy. Throughout the section we shall often omit some tensor factors of $\Id$ for brevity (e.g.~for an operator $M$ on Alice's system and an operator $N$ on Bob's system, $MN$ implicitly means $(M \otimes \Id)(\Id \otimes N)$). 

Without loss of generality, we assume the state shared by the two parties is pure, by giving the dishonest party the purification of any mixed state.
We shall also henceforth assume that all measurements are projective: in the case of cheating Alice, our analysis involves a standard Bell-nonlocality scenario, in which case it is a known result (see Appendix~\ref{app:projs} for a brief discussion) that measurements can be assumed to be projective. For cheating Bob, however, some of the measurements are sequential and hence more care is needed, since it was found in~\cite{BMKG13} that there exist sequential correlations that \emph{cannot} be described by projective measurements. (It would be convenient if Protocol~\ref{prot:XOTMS} did not involve Alice sending her input to Bob, thereby more closely resembling a standard nonlocal game. However, it would seem this step was necessary to allow an honest Bob to know which bit of his output he should use, as previously mentioned regarding Protocol~\ref{prot:XOTMS}.) We prove that the measurements can also be assumed to be projective in the case of cheating Bob in Lemma~\ref{lemm:proj} in Appendix~\ref{app:projs}.
 
\paragraph{Cheating Bob.}
Bob never calls for a test if he is dishonest, since if he does so then he must either abort or perform a protocol perfectly secure against him. Hence we can consider only two families of measurement choices that dishonest Bob may perform --- one family that is used if Alice calls for a test, and one family that is used in the subroutine running Protocol~\ref{prot:XOTMS}. As noted previously, dishonest Bob may not be physically supplying the test input Alice provides to his box and reporting the corresponding output. However, it is still true that the process of him producing a reply $b^A$ to Alice's message $y^A$ can be modelled as a measurement process described by some operators acting on his share of the state, and it is these operators that we will be calling $N^\test_{y^A|b^A}$ in the proof. (An alternative perspective is that any operation dishonest Bob might perform on his share of the state might as well have been ``encoded'' in the box directly.) Still, we highlight that it is important here that $y^A$ was chosen by honest Alice rather than self-reported by dishonest Bob --- this ensures that it was indeed uniformly chosen.

Since Alice is honest, the possible measurements performed by her box are the same in both situations, although Bob could have preprogrammed what measurement is performed by Alice's box for a given input. We have used $a^A, b^A, x^A, y^A$ to denote the inputs and outputs when Alice runs the test; here we shall use $a, b, x, y$ to denote the variables when the Protocol~\ref{prot:XOTMS} subroutine is run. 

Let Alice and Bob's shared pure state be $\ket{\rho}$, and let $\rho = \kb{\rho}$. Let Alice's measurement operators on her system for input $a$ and output $x$ be denoted by ${\meas}_{x|a}$ and Bob's test measurement operators on his system for input $b^A$ and output $y^A$ be denoted by $N^\test_{y^A|b^A}$. 
The probability of outputs $x^A, y^A$ on inputs $a^A, b^A$ by Alice during the test is given by 
$\norm{\meas_{x^A|a^A}N^\test_{y^A|b^A}\ket{\rho}}_2 = \tr(\meas_{x^A|a^A}N^\test_{y^A|b^A}\rho)$,
since we assumed the measurements are projective. Bob's cheating probability when Alice runs the test is just his probability of passing the test, since he can guess $x$ perfectly if he passes the test. Since honest Alice picks the testing inputs $a^A, b^A$ uniformly, Bob's probability of passing the test is given by
\begin{equation}\label{eq:bob-test}
P^\test_B = 
\frac{1}{9}\sum_{a^A, b^A} \, \sum_{x^A, y^A : x^A_{b^A} = y^A_{a^A}}\tr(\meas_{x^A|a^A}N^\test_{y^A|b^A}\rho).
\end{equation}

Let Bob's measurement operators when Alice does not run the test be denoted by $N^\prot$. There are two kinds of measurement operators here --- Bob performs a measurement with no input to generate the $y$ that he sends to Alice in step~\ref{step:bob-y} of the Protocol~\ref{prot:XOTMS} subroutine; these measurement operators will be denoted by $N^\prot_y$. The subnormalized shared state after Bob performs this measurement is given by $N^\prot_y \rho N^\prot_y$. After receiving Alice's input $a$ as a message from Alice, Bob performs another measurement to produce his guess $g$ for Alice's output $x$. Since the only information he learns from Alice is $a$, he uses this as his input and applies a measurement operator $N^\prot_{g|a}$. The probability of Alice outputting $x$ on input $a$ and Bob outputting $g$ on input $a$ and sending $y$ at the beginning is given by 
$\norm{\meas_{x|a}N^\prot_{g|a}N^\prot_y \ket{\rho}}_2 = \tr(\meas_{x|a}N^\prot_{g|a}N^\prot_y\rho N^\prot_y)$.
Alice picks her input $a$ uniformly at random and Bob cheats successfully in this case if $y, a$ and $x$ are compatible (in the Protocol~\ref{prot:XOTMS} subroutine) and his guess satisfies $g = x$. Hence Bob's cheating probability is given by
\begin{align}\label{eq:bob-prot}
P^\prot_B = 
& \frac{1}{3}\tr\bigg(\meas_{00|0}N^\prot_{00|0}(N^\prot_{00}\rho N^\prot_{00} + N^\prot_{01}\rho N^\prot_{01})  + \meas_{00|1}N^\prot_{00|1}(N^\prot_{00}\rho N^\prot_{00} + N^\prot_{10}\rho N^\prot_{10}) \nonumber \\
& + \meas_{00|2}N^\prot_{00|2}(N^\prot_{01}\rho N^\prot_{01} + N^\prot_{10}\rho N^\prot_{10}) + \sum_{a}\sum_{x \neq 00}\sum_y \meas_{x|a}N^\prot_{x|a}N^\prot_y\rho N^\prot_y \bigg). 
\end{align}
Using the fact that $\meas_{x|a}$ and $N^\prot_y$ act on different systems, the terms in the above expression can be rewritten in the form $\tr(\meas_{x|a}N^\prot_y N^\prot_{x|a}N^\prot_y \rho)$, which are terms in the NPA matrix.
Bob's overall cheating probability is given by the maximum value (over all possible states and measurements as described above) of
\begin{equation}\label{eq:bob-DI} 
q^A P^\test_B + (1-q^A) P^\prot_B.
\end{equation}

\paragraph{Cheating Alice.} Here we similarly assume dishonest Alice has measurement operators ${\meas}^\prot$ that are applied in her box when the Protocol~\ref{prot:XOTMS} subroutine is run, and measurement operators ${\meas}^\test$ that are applied when Bob calls for the test. Bob's measurement operators $N_{y|b}$ are the same throughout. As in Bob's case, Alice's cheating probability when Bob runs the test is given by
\begin{equation} 
P^\test_A =\frac{1}{9}\sum_{a^A, b^A} \, \sum_{x^A, y^A : x^A_{b^A} = y^A_{a^A}}\tr(\meas^\test_{x^A|a^A}N_{y^A|b^A}\rho).
\end{equation}

The only information Alice gets from Bob in the Protocol~\ref{prot:XOTMS} subroutine is Bob's output $y$. She uses this to choose a measurement to perform on her state, which produces an output $g$ that is her guess for Bob's input. This corresponds to a measurement operator $\meas^\prot_{g|y}$, and the probability of Bob producing output $y$ on input $b$ and Alice producing output $g$ on message $y$ from Bob is given by $\tr(\meas^\prot_{g|y}N_{y|b}\rho)$. Since honest Bob chooses his input $b$ uniformly, and Alice cheats successfully if her guess satisfies $g=b$, her cheating probability in this case is given by
\begin{equation} 
P^\prot_A = \frac{1}{3}\sum_{b}\sum_y\tr(\meas^\prot_{b|y}N_{y|b}\rho).
\end{equation}
Alice's overall cheating probability is given by the maximum value of
\begin{equation}\label{eq:alice-DI} 
q^B P^\test_A + (1-q^B) P^\prot_A.
\end{equation}

\paragraph{Using the NPA hierarchy.} The bound on $\PA$ in Theorem~\ref{th:ms-dixot} was computed by using NPA local level 1 (see Appendix~\ref{NPA}) to optimize Eq.~\eqref{eq:alice-DI}.\footnote{In principle, it could have been simpler to compute a bound using only the no-signalling constraints, as was done in part of the security proof in~\cite{SCA+11}. However, here we encountered the issue that non-signalling strategies that perfectly win the magic square game are not unique, and can even produce deterministic outcomes for one input pair.}
For $\PB$, we need to optimize Eq.~\eqref{eq:bob-DI}, and higher-order moments are required due to the sequential measurements. Hence we choose the index set $\ops$ (see Appendix~\ref{NPA}) to be 
\begin{equation}
\ops = (\{\Id\} \cup \{M_{x|a}\}) \circ (\{\Id\} \cup \{N^\test_{y|b}\} \cup \{N^\prot_{g|m}\} \cup \{N^\prot_{y}\} \cup \{N^\prot_{g|m} N^\prot_{y}\} ),
\end{equation}
where $\circ$ means all pairwise word concatenations between the two sets in order (and $\Id$ is the empty word). However, this matrix is very large if measurement operators for all outcomes are included, and hence we omitted one 
projector
per measurement by using the relation $\sum_x M_{x|a} = \sum_y N_{y|b} = \Id$ 
to eliminate $M_{11|a}, N_{11|b}$ from the objective function. Omitting these terms from the index set $\ops$ as well, its size is reduced to a more manageable $|\ops| = (1+9)(1+9+9+3+27) = 490$.

The results from solving these SDPs numerically are stated in Theorem~\ref{th:ms-dixot}.

\section{
Analytically bounding the cheating probabilities away from $1$ for Protocol~\ref{prot:di-xot}}
\label{sect:continuity}

We shall show the security proof for cheating Bob only (the proof for cheating Alice is analogous, though simpler because there are no sequential measurements by Alice). We start by first proving the following lemma, 
which formalizes the intuition that in order to show $\PB$ is bounded away from $1$ in Protocol~\ref{prot:di-xot}, it would be sufficient to argue that if Bob passes Alice's test with high probability, then he cannot cheat perfectly in the ``protocol branch''.
\begin{lemma}\label{lemm:gapped}
Consider any $q^A \in (0,1)$. Suppose there exists some $\eps > 0$ such that the value of
\begin{equation*}
\begin{aligned}
\text{maximize:}   & \quad P^\prot_B \\
\text{subject to:} & \quad P^\test_B \geq 1-\eps
\end{aligned}
\end{equation*}
is upper-bounded by $1-\eps/(1-q^A)$, where $P^\prot_B, P^\test_B$ are defined as in Eqs.~\eqref{eq:bob-test}--\eqref{eq:bob-prot}, and the supremum is taken over all states and measurements. Then for Protocol~\ref{prot:di-xot} performed with that choice of $q^A$, we have 
\begin{equation*}
\PB \leq 1-q^A \eps.
\end{equation*}
\end{lemma}
\begin{proof}
Suppose to the contrary that $\PB > 1-q^A \eps$. Then by Eq.~\eqref{eq:bob-DI}, this must be achieved by Bob implementing a strategy such that
\begin{equation}
P^\test_B  > \frac{1-q^A \eps - (1-q^A) P^\prot_B}{q^A} \geq \frac{1-q^A \eps - (1-q^A)(1)}{q^A} = 1-\eps.
\end{equation}
By hypothesis, this implies that $P^\prot_B \leq 1-\eps/(1-q^A)$. But this would imply
\begin{equation}
P^\xot_B \leq q^A(1) + (1-q^A) \left(1-\frac{\eps}{1-q^A}\right) = 1 - \eps < 1-q^A \eps,
\end{equation} 
contradicting the supposition that $\PB > 1-q^A \eps$.
\end{proof}

We now use the rigidity properties of the magic square game to prove (for any fixed choice of $q^A$) there indeed exists some $\eps > 0$ fulfilling the conditions of  
this lemma, 
and thus $\PB$ is bounded by some constant strictly less than 1. To do so, we first note that our device-dependent arguments directly yield the following lemma, which we shall use later. 
  
\begin{lemma}\label{lemm:BcheatDD}
Let $\ket{\Psi^\ms}$ and $M^\ms_{x|a}$ be the ideal state and measurements (for Alice) in the magic square game. 
If Alice and Bob share a state of the form $\ket{\rho} = \ket{\Psi^\ms}_{\calX_0\calX_1\calY_0\calY_1} \otimes \ket{\jk}_{\calJ_A\calJ_B}$, and Alice's measurements on $\calX_0\calX_1\calJ_A$ are of the form $M_{x|a} = M^\ms_{x|a} \otimes \Id_{\calJ_A}$, then for any measurements ${N}^\prot_{g|m},{N}^\prot_y$ by Bob on $\calY_0\calY_1\calJ_B$, we have
\begin{equation*}
P^\prot_B \leq 0.93628,
\label{eq:BcheatDD}
\end{equation*}
where $P^\prot_B$ is defined as in Eq.~\eqref{eq:bob-prot}. 
\end{lemma}
\begin{proof}
Since Alice's measurements act as the identity on $\calJ_A$, we may as well suppose Bob holds that system. Recall that our device-dependent arguments allowed the cheating party to perform arbitrary operations on their state. In particular, cheating Bob could append an ancilla $\ket{\jk}_{\calJ_A\calJ_B}$ and perform a joint measurement across this ancilla and his part of the state. Hence our bounds for the trusted-state-and-measurements scenario apply to the situation described here.
\end{proof}
 
Consider an arbitrary $\eps>0$, and consider any cheating strategy for Bob in Protocol~\ref{prot:di-xot} that achieves $P^\test_B \geq 1-\eps$. Then by Lemma~\ref{lemm:rigidity}, we have 
\begin{equation} \label{eq:rigidA}
\norm{(V_A M_{x|a} \otimes V_B) \ket{\rho} - ((M^\ms_{x|a} \otimes \Id) \ket{\Psi^\ms})\otimes\ket{\jk}}_2 \leq O(\eps^{1/4}).
\end{equation}
Before proceeding further, we give an informal overview of our subsequent argument. 
Qualitatively, Eq.~\eqref{eq:rigidA} means the (subnormalized) true post-Alice-measurement states are close to the ideal post-Alice-measurement states of the magic square game, up to local isometries/ancillas. We shall hence argue that after Bob performs his true measurements $N^\prot_{g|m}, N^\prot_y$ on the true post-Alice-measurement states, the resulting states are close (up to local isometries/ancillas) to those that would result from Bob performing some {other} measurements $\tilde{N}^\prot_{g|m}, \tilde{N}^\prot_y$ on the ideal post-Alice-measurement states, upon which our desired result quickly follows by observing that the latter is precisely a situation where Lemma~\ref{lemm:BcheatDD} applies. Roughly speaking, these measurements $\tilde{N}^\prot_{g|m}, \tilde{N}^\prot_y$ will be constructed from the true measurements $N^\prot_{g|m}, N^\prot_y$ by simply ``inverting'' the isometries mapping the true states to the ideal ones, giving us measurements that act on the ideal systems rather than the true ones. However, since isometries are not in general invertible, we shall need to account for some technical details while doing so.\footnote{An alternate perspective that may be helpful is that an isometry simply defines an embedding of Hilbert spaces, i.e.~we can choose bases such that $V\ket{\phi} = \ket{\phi} \oplus \mathbf{0}$ in explicit components. We can thus think of Eq.~\eqref{eq:rigidA} as the statement $\norm{((M_{x|a} \otimes \Id) \ket{\rho}) \oplus \mathbf{0} - ((M^\ms_{x|a} \otimes \Id) \ket{\Psi^\ms})\otimes\ket{\jk}}_2 \leq O(\eps^{1/4})$ in some choice of coordinates, in which case the subsequent construction is fairly intuitive.}

We now give the rigorous proof. As previously described, we take the input and output spaces of all the true measurement operators to be the same. In that case, we can define measurement operators $\tilde{N}^\prot_{y}, \tilde{N}^\prot_{g|m}$ on $\calY_0 \calY_1 \calJ_B$ via
\begin{equation}\label{eq:paddedops}
\tilde{N}^\prot_{y}=V_B N^\prot_{y}V_B^\dagger + \frac{\Pi_{V_B}}{2} \quad \text{and} \quad \tilde{N}^\prot_{g|m}=V_B N^\prot_{g|m}V_B^\dagger + \frac{\Pi_{V_B}}{2}
\end{equation}
where $\Pi_{V_B} = \Id_{\calY_0\calY_1\calJ_B} - V_B V_B^\dagger$. It is easy to verify that $\Pi_{V_B}$ is a projector and satisfies $V_B^\dagger \Pi_{V_B} = 0$, $\Pi_{V_B}V_B = 0$. 
The operators $\tilde{N}^\prot_{y}$ form a valid set of measurement operators on $\calY_0 \calY_1 \calJ_B$, since
\begin{equation}
\tilde{N}^{\prot \, \dagger}_{y}\tilde{N}^\prot_{y} = \left(V_B N^{\prot \, \dagger}_{y}V_B^\dagger + \frac{\Pi_{V_B}}{2}\right) \left(V_B N^\prot_{y}V_B^\dagger + \frac{\Pi_{V_B}}{2}\right) = V_B N^{\prot \, \dagger}_{y}N^\prot_{y}V_B^\dagger + \frac{\Pi_{V_B}}{4},
\end{equation}
which implies they sum to the identity on $\calY_0 \calY_1 \calJ_B$. Similarly, for each $m$, the operators $\tilde{N}^\prot_{g|m}$ form a valid set of measurement operators.

The ``sequential operators'' $\tilde{N}^\prot_{g|m}\tilde{N}^\prot_y$ thus also form a valid set of measurement operators, i.e.~we have $\sum_{gy} (\tilde{N}^\prot_{g|m}\tilde{N}^\prot_y)^\dagger\tilde{N}^\prot_{g|m}\tilde{N}^\prot_y = \Id_{\calY_0 \calY_1 \calJ_B}$ (for all $m$). This implies that for any (not necessarily normalized) $\ket{\nu} \in \calX_0\calX_1\calY_0\calY_1\calJ_A\calJ_B$, we have $\norm{(\Id_{\calX_0 \calX_1 \calJ_A} \otimes \tilde{N}^\prot_{g|m}\tilde{N}^\prot_y)\ket{\nu}}_2 \leq \norm{\ket{\nu}}_2$. Additionally, we note that
\begin{equation}
\tilde{N}^\prot_{g|m}\tilde{N}^\prot_{y} V_B =
\left(V_B N^\prot_{g|m} V_B^\dagger + \frac{\Pi_{V_B}}{2}\right) \left(V_B N^\prot_{y}V_B^\dagger + \frac{\Pi_{V_B}}{2}\right) V_B = V_B N^\prot_{g|m}N^\prot_y.
\end{equation}
Applying these properties to Eq.~\eqref{eq:rigidA}, followed by the reverse triangle inequality, we finally obtain the following (denoting $M^\ms_{x|a}\otimes\Id_{\calJ_A} = \tilde{M}^\ms_{x|a}$):
\begin{align}
O(\eps^{1/4}) &\geq \norm{(V_A M_{x|a} \otimes \tilde{N}^\prot_{g|m}\tilde{N}^\prot_y V_B) \ket{\rho} - (\tilde{M}^\ms_{x|a} \otimes \tilde{N}^\prot_{g|m}\tilde{N}^\prot_y) \ket{\Psi^\ms}\otimes\ket{\jk}}_2 \nonumber \\
&=\norm{(V_A M_{x|a} \otimes V_B N^\prot_{g|m}N^\prot_y) \ket{\rho} - (\tilde{M}^\ms_{x|a} \otimes \tilde{N}^\prot_{g|m}\tilde{N}^\prot_y) \ket{\Psi^\ms}\otimes\ket{\jk}}_2 \nonumber \\
&\geq \left|\, \norm{(V_A M_{x|a} \otimes V_B N^\prot_{g|m}N^\prot_y) \ket{\rho}}_2 - \norm{(\tilde{M}^\ms_{x|a} \otimes \tilde{N}^\prot_{g|m}\tilde{N}^\prot_y) \ket{\Psi^\ms}\otimes\ket{\jk}}_2 \,\right| \nonumber \\
&= \left|\, \norm{(M_{x|a} \otimes N^\prot_{g|m}N^\prot_y) \ket{\rho}}_2 - \norm{(\tilde{M}^\ms_{x|a} \otimes \tilde{N}^\prot_{g|m}\tilde{N}^\prot_y) \ket{\Psi^\ms}\otimes\ket{\jk}}_2 \,\right|.
\label{eq:closeprobs}
\end{align}
Recall that $P^\prot_B$ is a sum of terms of the form $\norm{(M_{x|a} \otimes N^\prot_{g|m}N^\prot_y) \ket{\rho}}_2$. Let $\tilde{P}^\prot_B$ be the value obtained by replacing these terms with $\norm{(\tilde{M}^\ms_{x|a} \otimes \tilde{N}^\prot_{g|m}\tilde{N}^\prot_y) \ket{\Psi^\ms}\otimes\ket{\jk}}_2$ instead. Since this is the sum of finitely many terms, Eq.~\eqref{eq:closeprobs} implies that $P^\prot_B \leq \tilde{P}^\prot_B + O(\eps^{1/4})$. In addition, Lemma~\ref{lemm:BcheatDD} implies that $\tilde{P}^\prot_B \leq 0.93628$. Hence we conclude that by choosing sufficiently small $\eps > 0$, we would have $P^\prot_B \leq 0.93628+O(\eps^{1/4}) \leq 1-\eps/(1-q^A)$, which satisfies the conditions of Lemma~\ref{lemm:gapped}.

\begin{remark}
In principle, this approach can yield analytical bounds on $\PB$ by explicitly tracking the constants in the $O(\eps^{1/4})$ bound. However, the results of \cite{WBMS16} show that this $O(\eps^{1/4})$ self-testing bound becomes trivial at very small values of $\eps$ (less than $10^{-4}$, for our definition of $\eps$). 
We found that if the optimization in Lemma~\ref{lemm:gapped} is solved numerically using the NPA hierarchy instead of using this analytical bound, the results appear much more robust --- for instance, at $\eps=10^{-3}$ we still get a nontrivial bound of $P^\prot_B \leq 0.97365
$. Hence this indicates that the analytical approach here yields bounds that are quite far from tight. (Solving the analogous optimization for cheating Alice yields tighter bounds as compared to cheating Bob; e.g.~$P^\prot_A \leq 0.85855
$ at $\eps=10^{-3}$ and $P^\prot_A 
\leq 0.83590
$ at $\eps=10^{-5}$.)
\end{remark} 

\section{Discussion}

The protocol in this work is a somewhat ``minimal'' construction, in that it only involves a small number of rounds, and each check that Alice and/or Bob conducts appears critical for achieving any nontrivial bound on the cheating probabilities. However, a natural question to investigate in future work would be whether having more rounds of communication or larger entangled states could yield better protocols. For instance, our protocol yields different cheating probabilities for Alice and Bob, but it seems unclear whether it is simply due to the asymmetric nature of the task, or whether having further checks by each party could result in more symmetric bounds.
As for trying to achieve better security by using more states, there are various self-testing results for devices playing many parallel instances of nonlocal games~\cite{McKague16,McKague17,CN16,Col17}, although it is not immediately clear whether they are sufficiently robust to be converted into a protocol of this form. Another approach could be to design protocols based on sequential rounds of nonlocal games, such as in~\cite{AMPS16,KW16}. 

It may also be interesting to try converting the XOT protocol in this work into a protocol for OT, or other tasks such as bit commitment and coin flipping as mentioned in the introduction. As a starting point, it is fairly straightforward to design a conversion to OT that maintains perfect completeness (in the following, we use $(x'_0,x'_1)$ to denote Alice's final OT output and $(b',y')$ to denote Bob's): 
\begin{protocol}[Converting XOT to OT] 
\label{prot:BBOT} 
\hfill
\begin{enumerate}
\item Alice and Bob perform an XOT protocol so that Alice has $(x_0,x_1,x_2)$ and Bob has $(b,y=x_b)$.
\item Bob picks $c \in \{0,1,2\}$ uniformly at random such that $c \neq b$, and sends it to Alice. 
\item \begin{itemize}
\item If $c=0$, Alice outputs $(x_0',x_1') = (x_1, x_0\oplus x_1)$, and Bob outputs $(b',y')=(b+2\mod 3,y)$.
\item If $c=1$, Alice outputs $(x'_0,x'_1) = (x_0\oplus x_1, x_0)$, and Bob outputs $(b',y')=(b+1 \mod 3,y)$.
\item If $c=2$, Alice outputs $(x'_0,x'_1) = (x_0,x_1)$, and Bob outputs $(b',y')=(b,y)$.
\end{itemize}
\end{enumerate} 
\end{protocol}
\noindent However, analyzing the soundness of the resulting OT protocol (using the protocol in this work to achieve the XOT part) would require further developing an appropriate SDP and/or self-testing analysis --- it does not seem straightforward to directly apply the final cheating-probability bounds for our XOT protocol, as they do not impose particularly detailed constraints on the nature of the dishonest party's side-information. We leave this for future work.

\section*{Acknowledgements}

We thank Jean-Daniel Bancal, Andrea Coladangelo, L{\'{\i}}dia del Rio, Honghao Fu, Anand Natarajan, Christopher Portmann, Xingyao Wu and Vilasini Venkatesh for helpful discussions. We also thank the reviewers for their feedback and suggestions for improvement.

Research at Perimeter Institute is supported in part by the Government of Canada through the Department of Innovation, Science and Economic Development Canada and by the Province of Ontario through the Ministry of Economic Development, Job Creation and Trade.

Research at the Institute for Quantum Computing is supported by Innovation, Science and Economic Development (ISED) Canada. S. K. and E. T. acknowledge support from the NSERC Canada Discovery Grants Program, and S. K. from Fujitsu Labs America.

Part of this work was done when S. K. was at the Centre for Quantum Technologies, National University of Singapore, Singapore (partially funded by the Singapore Ministry of Education and the National Research Foundation, Prime Minister's Office, Singapore), and E. T. was at the Institute for Theoretical Physics, ETH Z\"{u}rich, Switzerland (funded by the Swiss National Science Foundation via the National Center for Competence in Research for Quantum Science and Technology (QSIT), the Air Force Office of Scientific Research (AFOSR) via grant FA9550-19-1-0202, and the QuantERA project eDICT).  

\bibliographystyle{alpha}
\bibliography{bib-KST-dixot} 

\appendix 

\section{Background: The NPA hierarchy}  
\label{NPA}

In this appendix, we briefly explain the NPA hierarchy as used in Section~\ref{sect:DINPA}.
The NPA hierarchy~\cite{NPA08} is a method to characterize the set of \emph{quantum behaviours}, which are defined to be distributions that can be written as
\begin{equation}\label{eq:Cqc}
\Pr(xy|ab) = \bra{\nu} M_{x|a} N_{y|b} \ket{\nu} , 
\end{equation}
where $\ket{\nu}$ is a state in a Hilbert space $\calH$ and $\{M_{x|a}\}_{x\in\X,a\in\A}, \{N_{y|b}\}_{y\in\Y,b\in\B}$ are sets of operators on $\calH$ satisfying the following properties:
\begin{equation}\label{eq:commPVMs}
\begin{aligned}
&M_{x|a}^\dagger = M_{x|a} \text{ and } N_{y|b}^\dagger = N_{y|b}, \\
&M_{x|a} M_{x'|a} = \delta_{xx'} M_{x|a} \text{ and } N_{y|b} N_{y'|b} = \delta_{yy'} N_{y|b}, \\
&\sum_{x\in\X} M_{x|a} = \Id \text{ and } \sum_{y\in\Y} N_{y|b} = \Id \quad\forall a,b, \\
&[M_{x|a}, N_{y|b}] = 0 \quad\forall x,y,a,b.
\end{aligned}
\end{equation}
Such distributions can be produced, for instance, by projective measurements on the subsystems of a tensor product of two Hilbert spaces\footnote{While every distribution achievable on such a tensor-product space can easily be expressed in the form~\eqref{eq:Cqc}, the converse is not necessarily true. However, the former property implies the NPA hierarchy characterizes a superset of the distributions we consider, which is sufficient for our goal of upper bounding the cheating probabilities --- though the bounds may not be tight.}, as we are studying in this work. 
The NPA hierarchy consists of an infinite sequence of SDP conditions on $\Pr(xy|ab)$ which are all satisfied if and only if it is a quantum behaviour. 

Specifically, let $\ops$ denote any finite list of finite-length words formed by letters from the alphabet $\{M_{x|a}\}_{x\in\X,a\in\A} \cup \{N_{y|b}\}_{y\in\Y,b\in\B}$, 
and let $\ops_j$ denote the $j^\text{th}$ element of $\ops$. It is not hard to show that if we have a state and operators satisfying the conditions in Eq.~\eqref{eq:commPVMs}, then any $|\ops| \times |\ops|$ matrix $\Gamma$ with entries given by
\begin{equation}
\Gamma_{jk} = \bra{\nu} \ops_j^\dagger \ops_k \ket{\nu},
\end{equation}
interpreting the words as the corresponding operators on $\calH$, must be positive semidefinite. In particular, if $\Pr(xy|ab)$ is a quantum behaviour, then some of the entries of $\Gamma$ correspond to values of $\Pr(xy|ab)$, hence imposing an SDP constraint on this distribution. Additional constraints on $\Gamma$ can be imposed based on the relations in Eq.~\eqref{eq:commPVMs}. 

We hence see that any choice of $\ops$ yields an SDP characterization of a superset of the set of quantum behaviours (this is sufficient to yield upper bounds on the cheating probabilities we consider in this work). A sequence of increasingly tight bounds can be obtained via appropriate choices of $\ops$. Specifically, if $\ops$ is chosen to be the set of all words of length at most $k$, this will be referred to as \emph{global level $k$} of the NPA hierarchy. It was shown in~\cite{NPA08} that if a distribution $\Pr(xy|ab)$ satisfies the SDP constraint at global level $k$ for all $k$, then it must be a quantum behaviour. Hence the hierarchy describes a sequence of increasingly smaller sets that eventually converge to the set of quantum behaviours. 

As a slight variation, one can also consider \emph{local level $k$}, which refers to taking $\ops$ to be the set of all words with at most $k$ letters from $\{M_{x|a}\}_{x\in\X,a\in\A}$ and at most $k$ letters from $\{N_{y|b}\}_{y\in\Y,b\in\B}$. (Basically, local level $k$ defines a slightly smaller $\ops$ than global level $2k$.) Our results are based on local level 1 of the hierarchy or slight extensions thereof (see Section~\ref{sect:DINPA}). In addition, in our work we need to consider sequential measurements, which are not directly covered by the above description. However, as we discuss in Section~\ref{sect:DINPA}, the probabilities produced by such measurements can also be expressed as elements of the matrix $\Gamma$ and can be assumed projective, hence we can bound them using this approach. 

\section{Proof of Lemma~\ref{lemm:rigidity}} 
\label{sect:rigid-proof}

In this appendix we prove the rigidity lemma needed for the proof of our fully device-independent protocol. 
Again, in this section we often leave some tensor factors of $\Id$ implicit. 

Following \cite{CN16}, we define the following operators from the ideal measurements of the magic square game
\begin{equation}
\begin{gathered}
X^\ms_1 = \sum_{x_0x_1}(-1)^{x_1}M^\ms_{x_0x_1|1}, \quad \quad X^\ms_2 = \sum_{x_0x_1}(-1)^{x_0}M^\ms_{x_0x_1|1}, \\
Z^\ms_1 = \sum_{x_0x_1}(-1)^{x_0}M^\ms_{x_0x_1|0}, \quad \quad Z^\ms_2 = \sum_{x_0x_1}(-1)^{x_1}M^\ms_{x_0x_1|0}, \\
X^\ms_3 = \sum_{y_0y_1}(-1)^{y_1}N^\ms_{y_0y_1|1}, \quad \quad X^\ms_4 = \sum_{y_0y_1}(-1)^{y_1}N^\ms_{y_0y_1|0}, \\
Z^\ms_3 = \sum_{y_0y_1}(-1)^{y_0}N^\ms_{y_0y_1|0}, \quad \quad Z^\ms_4 = \sum_{y_0y_1}(-1)^{y_0}N^\ms_{y_0y_1|1}, \\
W^\ms_1 = \sum_{x_0x_1}(-1)^{x_0}M^\ms_{x_0x_1|2}, \quad \quad W^\ms_2 = \sum_{x_0x_1}(-1)^{x_1}M^\ms_{x_0x_1|2}, \\
W^\ms_3 = \sum_{y_0y_1}(-1)^{y_0}N^\ms_{y_0y_1|2}, \quad \quad W^\ms_4 = \sum_{y_0y_1}(-1)^{y_1}N^\ms_{y_0y_1|2}.
\end{gathered}
\end{equation}
$X_k, Z_k, W_k$ are defined similarly from the actual measurement operators $M_{x_0x_1|a}, N_{y_0y_1|b}$ (assumed projective, as declared in the Lemma~\ref{lemm:rigidity} conditions) used by Alice and Bob.

Both \cite{WBMS16} and \cite{CN16} use the following lemma due to McKague, which we shall also need.

\begin{lemma}[\cite{McKague16} Lemma~6]\label{lemm:mckg}
Suppose $\ket{\rho}$ is a state on $\calX\calY$ and $\{X_1, X_2\}$, $\{Z_1, Z_2\}$ are sets of Hermitian, unitary, commuting operators acting only on $\calX$, and $\{X_3, X_4\}$, $\{Z_3, Z_4\}$ are sets of Hermitian, unitary, commuting operators acting only on $\calY$, such that for all $k\neq l$,
\begin{align*}
& \norm{(X_kZ_l-Z_lX_k)\ket{\rho}}_2 \leq \delta \\
& \norm{(X_k-Z_{k+2})\ket{\rho}}_2 \leq \delta \\
& \norm{(X_kZ_k + Z_kX_k)\ket{\rho}}_2 \leq \delta.
\end{align*}
For $p, q \in \{0,1\}^4$, let $X^p = \prod_{i=1}^4X_i^{p_i}$ 
and let $Z^q, (X^\ms)^p, (Z^\ms)^q$ be defined similarly. Then, there exist isometries $V_A : \calX \to \calX_0\calX_1\calJ_A$ and $V_B : \calY \to \calY_0\calY_1\calJ_B$ such that for any $p, q \in \{0,1\}^4$,
\[ \norm{(V_A\otimes V_B) (X^pZ^q)\ket{\rho} - (X^\ms)^p (Z^\ms)^q\ket{\Psi^\ms}\otimes\ket{\jk}}_2 \leq O(\sqrt{\delta}).\]
\end{lemma}
Suppose the state $\ket{\rho}$ and measurements $M_{x|a}, N_{y|b}$ win the magic square game with probability $1-\eps$. Let $p_{ab}$ denote the probability that this strategy wins when Alice and Bob's inputs are $a$ and $b$. Since the inputs are uniformly distributed, this means that $\frac{1}{9}\sum_{a,b}p_{ab} = 1-\eps$. Since every probability is upper bounded by 1, for any fixed $a, b$ we have,
\begin{equation} p_{ab} = 9(1 - \eps) - \sum_{(a', b') \neq (a, b)}p_{a'b'} \geq 9(1- \eps) - 8 \geq 1 - 9\eps.\end{equation}
This implies that
\begin{equation}\label{eq:x1x3}
\langle \nu |X_1X_3|\nu\rangle 
= \sum_{x_0x_1y_0y_1} (-1)^{x_1+y_1} \langle \nu |M_{x_0x_1|1} N_{y_0y_1|1}|\nu\rangle 
= p_{11} - (1-p_{11})
\geq 1 - 18\eps.
\end{equation}
Similarly, using the lower bound on $p_{ab}$ for other values of $a, b$ we get
\begin{align}
\langle \rho |X_2X_4|\rho\rangle & \geq 1 - 18\eps \label{eq:x2x4}\\
\langle \rho |X_1X_2W_4|\rho\rangle & \geq 1 - 18\eps \label{eq:xw4} \\
\langle \rho |Z_1Z_3|\rho\rangle & \geq 1 - 18\eps \\
\langle \rho |Z_2Z_4|\rho\rangle & \geq 1 - 18\eps \\
\langle \rho |Z_1Z_2W_3|\rho\rangle & \geq 1 - 18\eps \label{eq:zw3} \\
-\langle \rho |W_1X_4Z_3|\rho\rangle & \geq 1- 18\eps \label{eq:w1zx} \\
-\langle \rho |W_2X_3Z_4|\rho\rangle & \geq 1 - 18\eps \label{eq:w2zx} \\
-\langle \rho |W_1W_2W_3W_4|\rho\rangle & \geq 1 - 18\eps. \label{eq:w-s}
\end{align}
Using these relations, we can show that $\norm{(X_1Z_1 + Z_1X_1)\ket{\rho}}_2 \leq O(\sqrt{\eps})$, following \cite{WBMS16,CN16}. 
(The argument basically relies on a sequence of bounds obtained in the same way as Eq.~\eqref{eq:normtrick} below.)
Similarly, we can show that all the conditions for Lemma~\ref{lemm:mckg} hold for the $X_k, Z_k$ operators, and hence the conclusion of the lemma holds with a bound of $O(\eps^{1/4})$. 

Now, consider the cases $p = 1000, 0100, 1100$ with $q$ being $0000$ in all cases in the statement of the lemma. Using the definitions of $X_1$ and $X_2$ we have,
\begin{alignat}{3}
& \Vert(V_A\otimes V_B)((M_{00|1} - M_{01|1} + M_{10|1} - M_{11|1})\otimes\Id)\ket{\rho} && \nonumber \\
& - ((M^\ms_{00|1} - M^\ms_{01|1} + M^\ms_{10|1} - M^\ms_{11|1})\otimes\Id)\ket{\Psi^\ms}\otimes\ket{\jk}\!\Vert_2 && \leq O(\eps^{1/4}) \label{eq:x1} \\
& \Vert(V_A\otimes V_B)((M_{00|1} + M_{01|1} - M_{10|1} - M_{11|1})\otimes\Id)\ket{\rho} && \nonumber \\
& - ((M^\ms_{00|1} + M^\ms_{01|1} - M^\ms_{10|1} - M^\ms_{11|1})\otimes\Id)\ket{\Psi^\ms}\otimes\ket{\jk}\!\Vert_2 && \leq O(\eps^{1/4}) \\
& \Vert(V_A\otimes V_B)((M_{00|1} - M_{01|1} - M_{10|1} + M_{11|1})\otimes\Id)\ket{\rho} && \nonumber \\
& - ((M^\ms_{00|1} - M^\ms_{01|1} - M^\ms_{10|1} + M^\ms_{11|1})\otimes\Id)\ket{\Psi^\ms}\otimes\ket{\jk}\!\Vert_2 && \leq O(\eps^{1/4}) \label{eq:x1x2}
\end{alignat}
where in the last inequality we have simplified $X_1X_2$ using the fact that the $M_{x_0x_1|1}, M^\ms_{x_0x_1|1}$ operators are orthogonal projectors. Moreover, both the $M_{x_0x_1|1}$ and $M^\ms_{x_0x_1|1}$ operators must sum to the identity on their respective spaces, which implies (via the $p = q = 0000$ case) that
\begin{alignat}{3}
& \Vert(V_A\otimes V_B)((M_{00|1} + M_{01|1} + M_{10|1} + M_{11|1})\otimes\Id)\ket{\rho} && \nonumber \\
& - ((M^\ms_{00|1} + M^\ms_{01|1} + M^\ms_{10|1} + M^\ms_{11|1})\otimes\Id)\ket{\Psi^\ms}\otimes\ket{\jk}\!\Vert_2 && \leq O(\eps^{1/4}). \label{eq:x_id}
\end{alignat}
Eqs.~\eqref{eq:x1}--\eqref{eq:x_id} involve linearly independent combinations of the $M_{x_0x_1|1}$ operators (and the corresponding ones for $M^\ms_{x_0x_1|1}$). Hence from linear combinations of these we can self-test the $M_{x_0x_1|1}$ operators, by applying the triangle inequality.

Similarly, $M^\ms_{x_0x_1|0}$ measurement can be self-tested by taking $p=0000$ with $q=1000, 0100, 1100$, and using the fact that $M^\ms_{x_0x_1|0}$ sum to identity. For self-testing $N^\ms_{y_0y_1|1}$ we take the equations for $(p=0010,q=0000), (p=0000,q=0001), (p=0010,q=0001)$, with the identity condition. And finally for $N^\ms_{y_0y_1|0}$ we take $(p=0001,q=0000), (p=0000, q=0010), (p=0001,q=0010)$ with the identity condition.

To self-test the $N_{y_0y_1|2}$ operators, note that from Eq.~\eqref{eq:xw4},
\begin{align}
\norm{(W_4 - X_1X_2)\ket{\rho}}_2^2 & = \langle\rho|W_4^2|\rho\rangle + \langle\rho|(X_1X_2)^2|\rho\rangle - \langle\rho|W_4X_1X_2|\rho\rangle - \langle\rho|X_1X_2W_4|\rho\rangle \nonumber \\ 
& = 2 - 2\langle\rho|X_1X_2W_4|\rho\rangle \leq 2 - 2(1-18\eps) = 36\eps,
\label{eq:normtrick}
\end{align}
where we have used the fact that $W_4$ and $X_1X_2$ square to identity by construction, and they also commute due to acting on only Alice and Bob's registers respectively. Using Eq.~\eqref{eq:zw3} similarly we have,
\begin{align}
\norm{(W_3 - Z_1Z_2)\ket{\rho}}_2 & \leq O(\sqrt{\eps}) \label{eq:w3z-norm}\\
\norm{(W_4 - X_1X_2)\ket{\rho}}_2 & \leq O(\sqrt{\eps}). \label{eq:w4x-norm}
\end{align}
Now, using the fact that $\norm{W_3}_\infty, \norm{X_1X_2}_\infty = 1$, and $W_3, W_4$ commute with $X_1, X_2, Z_1, Z_2$,
\begin{align}
\norm{(W_3W_4 - X_1X_2Z_1Z_2)\ket{\rho}}_2 & \leq \norm{W_3(W_4 - X_1X_2)\ket{\rho}}_2 + \norm{(W_3X_1X_2 - X_1X_2Z_1Z_2)\ket{\rho}}_2 \nonumber \\
 & \leq \norm{(W_4 - X_1X_2)\ket{\rho}}_2 + \norm{X_1X_2(W_3 - Z_1Z_2)\ket{\rho}}_2 \nonumber \\
 & \leq \norm{(W_4 - X_1X_2)\ket{\rho}}_2 + \norm{(W_3 - Z_1Z_2)\ket{\rho}}_2 \nonumber \\
 & \leq O(\sqrt{\eps}). \label{eq:w3w4}
\end{align}
We have already seen the self-testing of the $X_1X_2$ and $Z_1Z_2$ operators. Taking $(p=1100,q=1100)$ in Lemma~\ref{lemm:mckg} lets us self-test $X_1X_2Z_1Z_2$. Moreover, since the ideal magic square state and measurements win the magic square game perfectly, they must satisfy analogous relations with $\eps=0$, i.e.~we have
\begin{equation}
\begin{gathered}
Z^\ms_1Z^\ms_2\ket{\Psi^\ms} = W^\ms_3\ket{\Psi^\ms}, \quad \quad X^\ms_1X^\ms_2\ket{\Psi^\ms} = W^\ms_4\ket{\Psi^\ms}, \\
X^\ms_1X^\ms_2Z^\ms_1Z^\ms_2\ket{\Psi^\ms} = W^\ms_3W^\ms_4\ket{\Psi^\ms}.
\end{gathered}
\end{equation}
Using these, Eqs.~\eqref{eq:w3z-norm}--\eqref{eq:w3w4} and the fact that the $N_{y_0y_1|2}$ operators add to identity lets us self-test this measurement.

Finally, to self-test the $M_{x_0x_1|2}$ operators, we note that using Eqs.~\eqref{eq:w1zx}--\eqref{eq:w-s} we get,
\begin{align}
\norm{(W_1 + X_4Z_3)\ket{\rho}}_2 & \leq O(\sqrt{\eps}) \label{eq:w1xz-norm} \\
\norm{(W_2 + X_3Z_4)\ket{\rho}}_2 & \leq O(\sqrt{\eps}) \\
\norm{(W_1W_2 + W_3W_4)\ket{\rho}}_2 & \leq O(\sqrt{\eps}). \label{eq:w1w2}
\end{align}
The ideal magic square state and measurements satisfy
\begin{equation}
\begin{gathered}
X^\ms_4Z^\ms_3\ket{\Psi^\ms} = -W^\ms_1\ket{\Psi^\ms} \quad X^\ms_3Z^\ms_4\ket{\Psi^\ms} = -W^\ms_2\ket{\Psi^\ms} \\
W^\ms_3W^\ms_4\ket{\Psi^\ms} = - W^\ms_1W^\ms_2\ket{\Psi^\ms}.
\end{gathered}
\end{equation}
Using these, Eqs.~\eqref{eq:w1xz-norm}--\eqref{eq:w1w2}, the statement of Lemma~\ref{lemm:mckg} for $(p=0001,q=0010), (p=0010,q=0001)$, the self-testing of $W_3W_4$ that we have already seen, and the identity condition for $M_{x_0x_1|2}$, we can self-test this measurement. This completes the list of all self-testing statements in Lemma~\ref{lemm:rigidity}.

\section{Dilation to projective measurements}
\label{app:projs}

In this appendix, we prove 
the claim in Section~\ref{sect:DINPA} that we can assume the measurements are projective without loss of generality. 
As mentioned previously, such reductions are already known for standard Bell-nonlocality scenarios --- for instance, the proof for tensor product measurements is just a simpler version of the construction here. (The proof for commuting measurements is more complicated~\cite{PT15}, and we do not aim to consider that model in this work.) However, our proof here accommodates some forms of sequential measurements, as required for our protocol.

First, to handle a dilation of isometries to unitaries that we will need if the systems are infinite-dimensional, we prove a useful lemma based on the Sz.-Nagy dilation for an isometry. 
\begin{lemma}\label{lemm:sznagy}
Consider a set of isometries $V_j:\calH_0 \to \calH_j$ for $j=1,2,\dots,n$, and another set of isometries $V'_k:\calH_n \to \calK_k$ for $k=1,2,\dots,m$. 
Let $\calL = \bigoplus_{j=0}^n \calH_j \oplus \bigoplus_{k=1}^m \calK_k $, and view all $H_j$ and $K_k$ as subspaces of $\calL$. Then there exist unitaries $\{U_j\}_{j=1}^n$ on $\calL$ such that for any $\ket{\nu} \in \calH_0$, we have $V_j \ket{\nu} = U_j \ket{\nu} \in \calH_j$. There also exist unitaries $\{U'_k\}_{k=1}^m$ on $\calL$ such that for any $\ket{\mu} \in \calH_n$, we have $V'_k \ket{\mu} = U'_k \ket{\mu} \in \calK_k$.
\end{lemma}
\begin{proof}
For brevity, we only prove the case for $n=m=2$; the construction for more isometries is precisely analogous. For any isometry $V$ with some domain $\calH$, we define the operator $\Pi_{V} = \Id_{\calH} - V V^\dagger$, which can be easily shown to be a projector satisfying $V^\dagger\Pi_{V} = \Pi_{V}V = 0$. Let
\begin{equation}
\begin{gathered}
U_1 = \begin{pmatrix} 0_{\calH_0} & -V_1^\dagger & & \\ V_1 & \Pi_{V_1} & & \\ & & \Id_{\calH_2} & \\ & & & \Id_{\calK_1 \oplus \calK_2} \end{pmatrix}, \quad
U_2 = \begin{pmatrix} 0_{\calH_0} & & -V_2^\dagger & \\ & \Id_{\calH_1} & & \\ V_2 & & \Pi_{V_2} & \\ & & & \Id_{\calK_1 \oplus \calK_2} \end{pmatrix}, \\
U'_1 = \begin{pmatrix} \Id_{\calH_0\oplus\calH_1} & & & \\ & 0_{\calH_2} & -V'^\dagger_1 & \\ & V'_1 & \Pi_{V'_1} & \\ & & & \Id_{\calK_2} \end{pmatrix}, \quad
U'_2 = \begin{pmatrix} \Id_{\calH_0\oplus\calH_1} & & & \\ & 0_{\calH_2} & & -V'^\dagger_2 \\ & & \Id_{\calK_1} & \\ & V'_2 & & \Pi_{V'_2} \end{pmatrix},
\end{gathered}
\end{equation}
where blank blocks in the matrices are taken to be filled with zero operators.
A direct calculation shows that these operators are indeed unitary and satisfy the desired property, e.g. for $U_1$ we have
\begin{align}
U_1^\dagger U_1 &= 
\begin{pmatrix} 0_{\calH_0} & V_1^\dagger & & \\ -V_1 & \Pi_{V_1} & & \\ & & \Id_{\calH_2} & \\ & & & \Id_{\calK_1 \oplus \calK_2} \end{pmatrix} 
\begin{pmatrix} 0_{\calH_0} & -V_1^\dagger & & \\ V_1 & \Pi_{V_1} & & \\ & & \Id_{\calH_2} & \\ & & & \Id_{\calK_1 \oplus \calK_2} \end{pmatrix} \nonumber \\
&= \begin{pmatrix} 0_{\calH_0} + V_1^\dagger V_1 & 0_{\calH_0} + V_1^\dagger \Pi_{V_1} & & \\ 0_{\calH_0} + \Pi_{V_1} V_1 & V_1 V_1^\dagger + \Pi_{V_1} & & \\ & & \Id_{\calH_2} & \\ & & & \Id_{\calK_1 \oplus \calK_2} \end{pmatrix} \nonumber \\
&= \Id_{\calH_0 \oplus \calH_1 \oplus \calH_2 \oplus \calK_1 \oplus \calK_2}
\end{align}
and
\begin{equation}
U_1 \ket{\nu}= 
\begin{pmatrix} 0_{\calH_0} & -V_1^\dagger & & \\ V_1 & \Pi_{V_1} & & \\ & & \Id_{\calH_2} & \\ & & & \Id_{\calK_1 \oplus \calK_2} \end{pmatrix} 
\begin{pmatrix} \ket{\nu} \\ \mathbf{0}_{\calH_1} \\ \mathbf{0}_{\calH_2} \\ \mathbf{0}_{\calK_1 \oplus \calK_2} \end{pmatrix} =
\begin{pmatrix} \mathbf{0}_{\calH_0} \\ V_1 \ket{\nu} \\ \mathbf{0}_{\calH_2} \\ \mathbf{0}_{\calK_1 \oplus \calK_2} \end{pmatrix}.
\end{equation}
\end{proof}

\begin{remark}
The construction in the above proof easily generalizes in several ways, for instance if one allows some of the isometries to have the same codomain, or if another set of isometries could be applied after $V'_k$. 
\end{remark}

We now prove the claim that the measurements can be assumed projective without loss of generality, stated as the following lemma:
\begin{lemma} \label{lemm:proj}
For any state ${\rho}$ and POVMs attaining some value for Bob's cheating probability (as expressed in Eq.~\eqref{eq:bob-DI}), there is another state ${\bar{\rho}}$ and \emph{projective} measurements (PVMs) that attain the same value. An analogous statement holds for Alice's cheating probability (as expressed in Eq.~\eqref{eq:alice-DI}). 
\end{lemma}

\begin{proof}
We describe the proof for cheating Bob only (the proof for cheating Alice is analogous, though simpler since there are no sequential measurements in that case).

First, let us denote Alice and Bob's measurement operators as $M_{x|a}$, $N^\test_{y^A|b^A}$, $N^\prot_y$ and $N^\prot_{g|m}$, as described previously, but here we shall not assume they are projective. 
Note that this is intended to denote the Kraus operators of the measurements, rather than the POVM elements. More specifically, taking $M_{x|a}$ (for a fixed $a$) as an example, this denotes a set of operators satisfying $\sum_x M_{x|a}^\dagger M_{x|a} = \Id$, such that the (subnormalized) state after performing the measurement on a state $\rho$ and getting outcome $x$ is $M_{x|a} \rho M_{x|a}^\dagger$. (This is slightly more detailed than specifying POVM elements $\povm_{x|a} \succeq 0$ such that $\sum_x \povm_{x|a} = \Id$ and $P(x) = \tr(\povm_{x|a} \rho)$, because it incorporates an explicit description of the post-measurement state. This allows, for instance, defining $M_{x|a}$ such that the measurement outcome and choice of measurement is encoded into the post-measurement state, which is important for sequential measurements.) For projective measurements, the two notions coincide and hence we did not need to specify this distinction previously.

Without loss of generality, we can assume the codomains of all the measurement operators are the same (if they were different, we could simply view all of them as being embedded in a common Hilbert space; for instance the direct sum of all the Hilbert spaces\footnote{Technically, the operators $N^\prot_{g|m}$ must be modified to ensure the POVM elements still sum to the identity on this space, but this is fairly straightforward (similar to the construction in Eq.~\eqref{eq:dilatedPVMs}).}).
Let $\calY'$ and $\calY''$ be 4-dimensional Hilbert spaces with orthonormal bases $\{\ket{y'}\}$ and $\{\ket{y''}\}$. 
We first define the following isometries:
\begin{equation}
\begin{gathered}
V^\test_{b^A} \ket{\nu}_{\calY} = \sum_{y^A} \ket{y^A}_{\calY'} N^\test_{y^A|b^A} \ket{\nu}_{\calY} ,\\
V^\prot \ket{\nu}_{\calY} = \sum_y \ket{y}_{\calY'} N^\prot_y \ket{\nu}_{\calY},\\
V^\prot_{m} \ket{y'}_{\calY'}\ket{\nu}_{\calY} = \sum_g \ket{g \, y'}_{\calY'' \calY'} N^\prot_{g|m} \ket{\nu}_{\calY}.
\end{gathered}
\end{equation}
We apply Lemma~\ref{lemm:sznagy} with $V^\test_{b^A}$ being the isometries $\{V_j\}_{j=1}^{n-1}$, $V^\prot$ being the isometry $V_n$, and $V^\prot_{m}$ being the isometries $\{V'_k\}_{k=1}^{m}$. This gives us corresponding unitaries on a Hilbert space $\calL$, which we now take as the space in which all subsequent expressions are embedded. (In the finite-dimensional case, one can perform a simpler construction in which $\calY''\calY'\calY$ itself is the common Hilbert space, see e.g.~\cite{NC10}.)
These unitaries match the actions of the isometries on their respective subspaces, i.e.~$U^\test_{b^A} \ket{\nu}_{\calY} = V^\test_{b^A} \ket{\nu}_{\calY}$, $U^\prot \ket{\nu}_{\calY} = V^\prot \ket{\nu}_{\calY}$ and $U^\prot_{m} \ket{\mu}_{\calY'\calY} = V^\prot_{m} \ket{\mu}_{\calY'\calY}$. 

For each isometry $V$, let $\calL_{\setminus V}$ denote the subspace of $\calL$ given by omitting the codomain of $V$ from the direct sum used to construct $\calL$ (we introduce this technical detail just to ensure the PVM sets we construct below indeed sum to the identity on $\calL$). We then define projectors
\begin{equation}\label{eq:dilatedPVMs}
\begin{gathered}
\proj^\test_{y^A|b^A} = U^{\test \, \dagger}_{b^A} ((\kb{y^A}_{\calY'} \otimes \Id_{\calY}) \oplus \delta_{y^A,00} \Id_{\calL_{\setminus V^\test_{b^A}}} ) U^\test_{b^A} ,\\
\proj^\prot_y = U^{\prot \, \dagger} ((\kb{y}_{\calY'} \otimes \Id_{\calY}) \oplus \delta_{y,00} \Id_{\calL_{\setminus V^\prot}} ) U^\prot , \\
\proj^\prot_{g|m} = U^{\prot \, \dagger} U^{\prot \, \dagger}_{m} ((\kb{g}_{\calY''} \otimes \Id_{\calY' \calY}) \oplus \delta_{g,00} \Id_{\calL_{\setminus V^\prot_{m}}} ) U^\prot_{m} U^\prot ,
\end{gathered}
\end{equation}
where $\delta_{j,k}$ denotes the Kronecker delta. It is easily seen that these form valid PVM sets. Importantly, these projectors have the useful property (omitting the $\Id_{\calL_{\setminus V}}$ terms for brevity)
\begin{equation}
\proj^\prot_{g|m} \proj^\prot_y = U^{\prot \, \dagger} U^{\prot \, \dagger}_{m} (\kb{g}_{\calY''} \otimes \Id_{\calY' \calY}) U^\prot_{m} (\kb{y^A}_{\calY'} \otimes \Id_{\calY}) U^\prot.
\label{eq:seqprojs}
\end{equation}

We construct unitaries $U_a$ and projectors $\proj_{x|a}$ analogously for Alice, though for that we only need to introduce a single ancillary system $\calX'$. As for the state, we can simply take $\ket{\bar{\rho}} = \ket{\rho}$, though with the understanding that it is now embedded in a larger Hilbert space.
Using Eq.~\eqref{eq:seqprojs}, we now compute some probabilities resulting from these projective measurements: 
\begin{align}
\norm{\proj_{x^A|a^A} \proj^\test_{y^A|b^A} \ket{\rho}}_2 &= \norm{U^\dagger_{a^A} U^{\test \, \dagger}_{b^A} \left(\ket{x \, y^A} M_{x^A|a^A} N^\test_{y^A|b^A} \ket{\rho}\right)}_2 \nonumber \\
&= \norm{M_{x^A|a^A} N^\test_{y^A|b^A} \ket{\rho}}_2, \\
\norm{\proj_{x|a} \proj^\prot_{g|m}\proj^\prot_y \ket{\rho}}_2 
&= \norm{U^\dagger_a U^{\prot \, \dagger} U^{\prot\,\dagger}_{m} \left(\ket{x \, g \, y} M_{x|a} N^\prot_{g|m} N^\prot_y \ket{\rho}\right)}_2 \nonumber \\
&= \norm{M_{x|a} N^\prot_{g|m} N^\prot_y \ket{\rho}}_2, 
\end{align}
since the 2-norm is unitarily invariant. These probabilities include all terms appearing in the cheating probability\footnote{We have exploited the fact that in our scenario, the measurements $M^\prot_{g|m}$ are always only applied after the measurement $M^\prot_y$. If there had been another possible measurement (including a trivial measurement, i.e.~identity operation) to perform before measuring $M^\prot_{g|m}$, it might not always be possible to ensure that the equivalent of Eq.~\eqref{eq:seqprojs} holds for that measurement as well. (In particular, when applying the NPA hierarchy, the projective measurements we constructed may not give the same value for all the $\Gamma$ matrix elements as the original POVMs. However, this is not an issue because the reduction to projective measurements occurs \emph{before} applying the NPA hierarchy.)}, and are equal to those attained by the original POVMs on $\ket{\rho}$. 
\end{proof}

\end{document}